\documentclass[%
 reprint,
nofootinbib,
 amsmath,amssymb,
 aps,
]{revtex4-1}

\usepackage{graphicx}
\usepackage{hyperref}

\usepackage{amsthm}
\usepackage{enumitem}
\usepackage{xcolor} 
\usepackage[framemethod=tikz]{mdframed}

\definecolor{warning_bgcol}{RGB}{252,248,229}
\definecolor{warning_textcol}{RGB}{111,89,54}
\definecolor{warning_linecol}{RGB}{248,235,207}
\definecolor{danger_bgcol}{RGB}{239,223,222}
\definecolor{danger_textcol}{RGB}{128,60,57}
\definecolor{danger_linecol}{RGB}{230,205,209}
\definecolor{success_bgcol}{RGB}{224,237,216}
\definecolor{success_textcol}{RGB}{68,104,60}
\definecolor{success_linecol}{RGB}{218,232,201}
\definecolor{info_bgcol}{RGB}{220,237,246}
\definecolor{info_textcol}{RGB}{58,100,126}
\definecolor{info_linecol}{RGB}{196,231,240}

\mdfdefinestyle{box_style}{
skipabove=\baselineskip,
skipbelow=\baselineskip,
innertopmargin=\baselineskip,
innerbottommargin=\baselineskip,
innerleftmargin=.4\baselineskip,
innerrightmargin=.4\baselineskip,
splittopskip =1.5\baselineskip,
splitbottomskip =\baselineskip,
roundcorner=.4\baselineskip
}
\mdfdefinestyle{warning_style}{
style=box_style,
backgroundcolor=warning_bgcol,
linecolor=warning_linecol,
fontcolor=warning_textcol,
}
\mdfdefinestyle{success_style}{
style=box_style,
backgroundcolor=success_bgcol,
linecolor=success_linecol,
fontcolor=success_textcol,
}
\mdfdefinestyle{danger_style}{
style=box_style,
backgroundcolor=danger_bgcol,
linecolor=danger_linecol,
fontcolor=danger_textcol,
}
\mdfdefinestyle{info_style}{
style=box_style,
backgroundcolor=info_bgcol,
linecolor=info_linecol,
fontcolor=info_textcol,
}

\newmdenv[style=warning_style]{warning}
\newmdenv[style=info_style]{info}
\newmdenv[style=danger_style]{danger}
\newmdenv[style=success_style]{success}

\theoremstyle{remark}
\newtheoremstyle{mythmstyle}%
{0pt}
{0pt}
{}
{}
{\bf}
{.}
{.5em}
{}

\theoremstyle{mythmstyle}
\newtheorem{thm}{Theorem}
\newtheorem{lem}[thm]{Lemma}

\surroundwithmdframed[style=info_style]{thm}
\surroundwithmdframed[style=info_style]{lem}
\surroundwithmdframed[style=info_style]{cor}
\surroundwithmdframed[style=info_style]{prop}
\surroundwithmdframed[style=warning_style]{defn}

\newcommand{\skproof}{Sketch of proof}
\makeatletter
\renewenvironment{proof}[1][\proofname]{\par
\pushQED{\qed}%
\normalfont \topsep6\p@\@plus6\p@\relax
\trivlist
\item\relax
{\bf
#1\@addpunct{.}}\hspace\labelsep\ignorespaces
}{%
\popQED\endtrivlist\@endpefalse
}
\makeatother

\begin{document}

\title{Transversal gates and error propagation in 3D topological codes}
\author{H\'ector Bomb\'in}
\affiliation{PsiQuantum, Palo Alto}

\begin{abstract} 
I study the interplay of errors with transversal gates in 3D color codes, and introduce some new such gates. Two features of the transversal T gate stand out: (i) it naturally defines a set of correctable errors, and (ii) it exhibits a `linking charge' phenomenon that is of interest for a wide class of 3D topologically ordered systems.
\end{abstract}

\maketitle

\section{Introduction}

In order to perform quantum computations in a fault-tolerant manner information has to stay encoded during the computation~\cite{lidar:2013:quantum}. The encoding introduces redundancy that makes it possible to regularly extract the entropy created by noise. Encoding information, however, creates a new difficulty: the gates that comprise the computation have to be encoded too, \emph{i.e.} they have to map encoded states to encoded states, but at the same time they need to preserve the structure of noise, so that error correction can still be successfully carried out. Typically noise is expected to have a local structure, meaning that errors are more likely the less qubits they affect. This is expected due to the locality of physical interactions.

A key method to compute with encoded states is using transversal gates~\cite{lidar:2013:quantum}. These are circuits that operate separately on blocks of a few qubits each, typically one for each logical qubit involved in the gate. Their great advantage is that they do not propagate errors between blocks, thus preserving the locality of noise.

In some scenarios, however, noise can naturally adopt a non-local structure and yet be correctable. This is the case for thermal noise in known self-correcting systems~\cite{dennis:2002:tqm, bombin:2013:self}, and also under a form of single-shot error correction that is inspired by self-correcting systems~\cite{bombin:2015:single-shot}. When transversal gates are used in such contexts, it is not a priori clear whether the propagated noise (through the action of the gate) is compatible with error correction or, in the case of self-correction, whether the propagated excitations typically dissipate without giving rise to logical errors.

Tetrahedral codes~\cite{bombin:2007:3dcc, bombin:2015:gauge}, which are part of the color code family~\cite{bombin:2006:2dcc, bombin:2007:branyons, bombin:2007:3dcc, bombin:2015:gauge}, are a class of three-dimensional topological stabilizer codes prominent for their transversal gates. In particular the T gate, a key gate that enables universal computation when supplemented with Clifford operations, is transversal in tetrahedral codes.
This paper (i) introduces a new set of transversal gates for tetrahedral codes, and (ii) analyzes different aspects of error propagation for all their transversal gates, with an emphasis on the T gate since it has the richest structure.

The enlarged set of gates presented here is essential for colorful quantum computation, a set of quantum computation schemes based on tetrahedral codes~\cite{bombin:2018:colorful} that is both inspired by and well-tailored to photonic quantum computing~\cite{rudolph:2017:optimistic}.
In its simplest form, colorful quantum computation is a topological fault-tolerant scheme for three-dimensional architectures in which all logical operations are transversal%
\footnote{
Including measurements supplemented with global classical computation.
} 
~\cite{bombin:2018:colorful}. This scheme relies on a form of single-shot error correction, in particular of the kind based on self-correction. One of the results of this paper is that, for all transversal gates of tetrahedral codes, the propagation of the non-local noise created by single-shot error correction only gives rise to local noise, and is thus compatible with error correction (section~\ref{sec:propagation}). This ensures the fault-tolerance of this form of colorful quantum computation, and also of other previous schemes~\cite{bombin:2016:dimensional}.

Remarkably, colorful quantum computation, despite being based on three-dimensional codes, also comprises an scalable scheme for two-dimensional architectures. The scheme relies on `just-in-time' (JIT) decoding for single-shot error correction. In JIT decoding information is decoded as it becomes available, in contrast with the conventional approach to the same problem, where it is decoded as a whole. The need for JIT decoding stems from the emergence of causality constraints in mapping a 3D code to a 2D architecture, because time plays the role of one of the spatial dimensions. Since JIT decoding is limited in this way by causality, but conventional decoding is not, it should perform worse. One of the outcomes of the detailed study of error propagation under the transversal T gate is that the errors caused by using JIT decoding, rather than conventional decoding, can essentially be transformed into erasure errors (section~\ref{sec:erasure}). 

Another outcome of the study of the transversal T gate is that the propagation of errors has an interesting algebraic structure that gives rise to a natural set of correctable errors, \emph{i.e.} one that is not based on any arbitrary choices but rather is fixed by the structure of the gate. This is true of any code that implements the T gate in the same way as tetrahedral codes (section~\ref{sec:T}).

Last, but not least, the T gate also hides some new physics. Tetrahedral codes, as any other family of topological codes, have a condensed matter model counterpart. For tetrahedral codes the corresponding physical system can be regarded as both a string-net and a membrane-net condensate~\cite{bombin:2007:branyons} in which excitations carry topological charge / flux. The transversal T gate turns out to exhibit a  `linking charge' phenomenon by which flux excitations exchange charge according to how they are linked, in a topological sense (section~\ref{sec:linking}). This phenomenon is relevant to a large class of systems (appendix~\ref{sec:linking_general}). 

\begin{danger}
\center {\bf Notation} is listed in appendix~\ref{sec:notation}.
\end{danger}

\section{Transversal T gate}\label{sec:T}

The transversal implementation of the logical T gate in 3D color codes~\cite{bombin:2007:3dcc, bombin:2015:gauge} is a generalization of the approach introduced in~\cite{knill:1996:threshold}. This section explores the propagation of noise for  such transversal T gates. The problem turns out to have some remarkable algebraic structure that naturally defines a set of correctable errors. 

The proofs of the lemmas in this section are in appendix~\ref{sec:Tproof}. Section~\ref{sec:propagation_T} and appendix~\ref{sec:noise_T} specialize the results of this section to 3D color codes.

\subsection{Setting}\label{sec:setting}

Throughout this section, unless explicitly indicated otherwise, $S$ is the stabilizer group of a fixed stabilizer code. We assume the following properties%
\footnote{These axioms have not been chosen to be minimal, but simply to offer a practical abstraction of tetrahedral codes.}%
, modeled after tetrahedral color codes~\cite{bombin:2007:3dcc, bombin:2015:gauge}.

\begin{warning}
\begin{enumerate}
\item $S$ is a CSS code, i.e.\begin{equation}
S=S_XS_Z, \qquad
S_X\subseteq P_X, 
S_Z\subseteq P_Z,
\end{equation}

\item the code subspace is invariant under a transversal gate of the form\begin{equation}
U = \bigotimes_q T^{b_q},
\qquad b_q\equiv 1 \mod 2,
\end{equation}

\item there is a single logical qubit, and

\item $X$ and $Z$ undetectable errors are related as follows:\begin{equation}
\mathcal Z_Z(S)= \{ Z_{\alpha\cap\beta} \,|\, X_\alpha, X_\beta\in\mathcal Z_X(S)\}.
\label{axiom_code}
\end{equation}

\end{enumerate}
\end{warning}

\noindent For readability, below `logical operator' stands for 'non-trivial logical Pauli operator', i.e. an element of\begin{equation}
\mathcal Z(S)-S.
\end{equation}

\subsection{Invariants}

The following objects will play an important role below.

\begin{warning}

Given $X_\alpha$ we define the Pauli groups
\begin{align}
G_\alpha &:= S_Z\cdot\{Z_{\alpha\cap\beta} \,|\, X_\beta\in \mathcal Z(S)\},
\\
H_\alpha &:= S_Z \cdot \{ Z_{\alpha\cap\beta} \,|\, X_\beta\in S \},
\end{align}
the integer\begin{equation}
g(\alpha):= \sum_{q\in \alpha} b_q,
\end{equation}
and the set of Pauli operators (a coset of the quotient $P_Z/G_\alpha$)
\begin{multline}
E_\alpha :=
 \{z\in P_Z
\,|\,
\forall X_\gamma\in \mathcal Z(G_\alpha)
\\
(z,X_\gamma)=(-1)^{g(\gamma\cap \alpha)/2}
\}.
\end{multline}

\end{warning}

\noindent They are invariant in the following sense%
\footnote{The function $g$ is also invariant, see the proof of lemma~\ref{lem:invariance}.}.

\begin{lem}\label{lem:invariance}
For any $X_\alpha$, any $X_\beta\in S$ and any $X_\gamma\in \mathcal Z(S)$
\begin{align}
G_{\alpha+\beta}&=G_\alpha,
\\
H_{\alpha+\gamma}&=H_\alpha,
\\
E_{\alpha+\beta}&=E_\alpha.
\end{align}
\end{lem}

\subsection{Tolerable errors}

The following class of errors plays an essential role in the main result.

\begin{warning}

$X_\alpha$ is {\bf tolerable} (respect to $S$) if $G_\alpha$ contains no logical operators.

\end{warning}

\noindent Here is a useful characterization.

\begin{lem}\label{lem:tolerability}
$X_\alpha$ is tolerable if and only if any of the following holds:
\begin{enumerate}
\item
$H_\alpha=G_\alpha$,
\item
there exists a logical $X_\lambda$ such that $Z_{\alpha\cap\lambda}\in S,$
\item
there exists a logical $X_\lambda\in\mathcal Z(G_\alpha)$,
\item
$X_{\alpha+\lambda}$ is \emph{not} tolerable, given any logical $X_\lambda$.
\item
$X_{\alpha+\beta}$ is tolerable, given any stabilizer $X_\beta$.
\end{enumerate}
\end{lem}

\noindent For every syndrome there exists exactly one class (up to stabilizers) of tolerable errors with that syndrome. Thus tolerable errors form a set of correctable errors. Section~\ref{sec:tolerable_cc} compares, in the context of tetrahedral codes, this set of correctable errors with another one defined in purely topological terms.

\subsection{Error syndromes}\label{sec:T:syndromes}

An error syndrome is an irrep of a stabilizer group $S$. That is, a syndrome is a morphism
\begin{equation}
S \rightarrow \{1, -1\},
\end{equation}
e.g. an eigenvalue assignment. Here $S$ is a CSS code and it is convenient to consider separately $S_X$ and $S_Z$ syndromes. In this section only $X$-error syndromes are of interest, i.e. $S_Z$ syndromes. In particular, due to lemmas~\ref{lem:invariance} and~\ref{lem:tolerability} the following objects are well defined.

\begin{warning}

Given an $X$-error syndrome $\phi$ let, for any tolerable $X_\alpha$ with syndrome $\phi$,
\begin{align}
H(\phi)&:= H_\alpha,
\\
E(\phi)&:= E_\alpha.
\end{align}

\end{warning}

\subsection{Error propagation}

The stage is finally set to describe the propagation of Pauli errors under a transversal T gate. 
The transversal $T$ gate commutes with $Z$ errors, so that it suffices to consider how it propagates $X$ errors. Because the $T$ gate is non-Clifford, it will \emph{not} map $X$ errors to Pauli errors. A depolarization operation is required to ensure that a Pauli error propagates to a distribution of Pauli errors. We consider $\mathcal D_{S_X}$, which amounts to apply randomly an element of $S_X$, or equivalently to measure the check operators in $S_X$ and forget the readout.
This depolarization is rather natural, in the sense that it becomes immaterial if followed by an ideal syndrome extraction.

\begin{thm}\label{thm:T}
For any $x\in P_X$ with syndrome $\phi$ and any encoded state $\rho$ 
\begin{equation}\label{eq:main}
(\mathcal D_{S_X}\circ \hat U\circ \hat x) (\rho) = 
(\hat x \circ \mathcal D_{E(\phi)}\circ \hat U \circ \hat w)(\rho),
\end{equation}
where $w =\mathbf 1$ if $x$ is tolerable, and otherwise $w$ is an encoded

\begin{equation}
U_0^\dagger X U_0 X
\end{equation}
gate, given that $U$ is an encoded $U_0$ gate.
\end{thm}

\noindent For tolerable errors $x$, the final state is subject, on top of the original error $x$, to a random error from the set $E(\phi)\subseteq P_Z$. Since the set $E(\phi)$ is a coset of the group $H(\phi)$, the error is only random up to an element of $H(\phi)$.

The central relation \eqref{eq:main} reveals that the notion of tolerability emanates from the transversal gate $U$. Remarkably, \emph{such a transversal gate naturally defines a set of correctable errors}.

\subsection{Error factorization}\label{sec:factorization}

Theorem~\ref{thm:T} is sometimes more useful if the error $x$ is factorized in the right way. We use additive notation for the abelian group of $X$-error syndromes.

\begin{warning}

The $X$-error syndromes $\phi_i$ are {\bf separated} (respect to $S$) if\begin{equation}
H\left(\sum_i \phi_i\right) = \prod_{i} H(\phi_i).
\end{equation}

\end{warning}

\begin{lem}\label{lem:factorization}

Given tolerable $X_{\alpha_i}$ with separated syndromes $\phi_i$, and\begin{equation}
\alpha =\sum_i \alpha_i,\qquad \phi=\sum_i \phi_i,
\end{equation}
\begin{enumerate}
\item
$X_\alpha$ is tolerable, and
\item
for any $z_i\in E(\phi_i)$
\begin{equation}
\prod_i z_i \prod_{i\neq j} Z_{\alpha_i\cap\alpha_j}\in E(\phi).
\end{equation}
\end{enumerate}
\end{lem}

\noindent The following criteria is useful in practice.

\begin{lem}\label{lem:separation}
Given $X_{\alpha_i}$, if $S_X$ has generators $X_{\beta_j }$ such that for each $j$ there exists at most a single $i$ satisfying
\begin{equation}
Z_{\alpha_i\cap\beta_j}\not\in S,
\end{equation}
then the operators $X_{\alpha_i}$ have separated syndromes.
\end{lem}

\subsection{Check operator erasure}\label{sec:erasure}

Suppose that the transversal gate $U$ and the depolarization operator $\mathcal D_{S_X}$ are applied to an encoded state subject to a \emph{known} tolerable error $x\in P_X$ with syndrome $\phi$. Unlike in the case of transversal Clifford gates, the error $x$ propagates to a random Pauli error. In particular, the propagated error is the combination of a known Pauli error and a random error in $H(\phi)$.

In the absence of further errors the randomness is not a problem: $H(\phi)$ contains no logical operators. In a fault-tolerant scenario, however, there will be further errors, and the randomness has an effect similar to the erasure of physical qubits in the code: some information is known to disappear. Check operators not commuting with $H(\phi)$ should be ignored.

To be more specific, suppose that the original error $X_{\alpha+\omega}$ is unknown but a \emph{good guess} $X_\alpha$ exists for it%
\footnote{
This scenario is realized in~\cite{bombin:2018:colorful}, where the error introduced in the preparation of encoded $X$ eigenstates using JIT decoding is guessed via conventional decoding at a later stage.
}. Assume that both $X_\alpha$ and $X_{\alpha+\omega}$ are tolerable. Since the final state is expected to be subject to a random error in $E_\alpha$, the natural approach to correct it is to
\begin{enumerate}
\item
apply some element of $E_\alpha$ to the final state,
\item
perform error correction utilizing only the syndrome of the subgroup
\begin{equation}
\mathcal Z_S(H_\alpha),
\end{equation}
\item 
using the full syndrome of $S$, apply any element of $H_\alpha$ that takes the resulting state back to an encoded state.
\end{enumerate}
Assuming this procedure, what is the effective noise? By lemma~\ref{lem:A6} below, on step 2 the effective error takes the form
\begin{equation}
Z_{\alpha\cap\omega}e_\omega, \qquad e_\omega\in E_\omega,
\end{equation}
up to an irrelevant element of $H_\alpha$.

\begin{lem}\label{lem:A6}
For any $X_\alpha, X_\omega$, and any $e_\sigma\in E_\sigma$, where $\sigma
\in\{\alpha,\omega,\alpha+\omega\}$, in coset notation,
\begin{equation}
e_{\alpha}e_{\alpha+\omega} G_\alpha G_{\alpha+\omega} = e_\omega Z_{\alpha\cap\omega}G_\alpha G_{\omega}.
\end{equation}
\end{lem}

\section{Confined error syndromes}\label{sec:confined}

This section discusses non-local error distributions that appear as residual noise when performing single-shot error correction in 3D topological codes~\cite{bombin:2015:single-shot}. The propagation of such errors under transversal gates is the subject of section~\ref{sec:propagation}, in the specific case of 3D color codes.

\subsection{Local noise}

In the analysis of fault tolerance it is useful to consider stochastic noise models in which our target computational state is afflicted by a distribution of errors. Typically, as the computation proceeds, the error distribution stays within a given family  of error distributions with high probability.

Local noise is the most important class of such families of distributions. 
\begin{warning}
A distribution of (error) operators, each with support on a set of qubits $A$, is {\bf local} with error rate $0\leq p<1$ if for any set of qubits $B$
\begin{equation}\label{eq:local}
\text{prob}(B\subseteq A)\leq p^{|B|}.
\end{equation}
\end{warning}

\subsection{Single-shot error correction}

There is no reason to consider only local noise. Indeed, going beyond local noise can lead to interesting new fault tolerant schemes. This is exemplified by single-shot error correction~\cite{bombin:2015:single-shot}, which in some cases relies on non-local forms of noise.

The essential idea behind single-shot error correction can be summarized in the equation
\begin{equation}
\partial^2=0.
\end{equation}  
Here the symbol $\partial$ represents syndrome operators. In particular, two different syndrome operators. For a stabilizer code, the first maps Pauli errors to their syndrome and the second maps noisy stabilizer outcomes to their syndrome (stemming from the redundancy of the measurements). Thus the relation simply says that a noiseless syndrome of a Pauli error has trivial syndrome. 

Single-shot error correction makes it possible to correct the errors of a noisy syndrome (to a certain degree). In some cases the price to pay is an unconventional form of residual noise.

\subsection{Syndrome confinement in 3D topological codes}

Consider 3D topological stabilizer codes with a membrane condensate picture, such as 3D toric codes or 3D color codes. In these codes there are two very different kind of error syndromes, which can be point-like or string-like (e.g. the syndrome of $Z$ or $X$ Pauli errors, respectively, in 3D color codes). Errors with a string-like syndrome can be corrected in a single-shot fashion, at a cost: residual errors no longer follows a local distribution, but rather the syndromes do, as explained next~\cite{bombin:2015:single-shot}.

The motivation for the new class of error distributions comes from the Hamiltonian picture, where syndromes are excitations: below a critical temperature strings are confined and the system is partially self-correcting. Inspired by this phenomena, in the fault-tolerant picture we similarly require that string-like syndromes follow a local distribution, given some set of localized check operators that generate the stabilizer group.

\begin{warning}
A distribution of error syndromes, each with negative eigenvalue for a set $A$ of check operators, is {\bf local} with `syndrome rate' $0\leq p<1$ if for any set of check operators $B$ the inequality~\eqref{eq:local} holds.
\end{warning}
\noindent
Below a critical syndrome rate the strings are confined: the probability to find strings of a given length at a given location decreases exponentially with length. This is the regime of interest.

\begin{figure}
\includegraphics[width=\columnwidth]{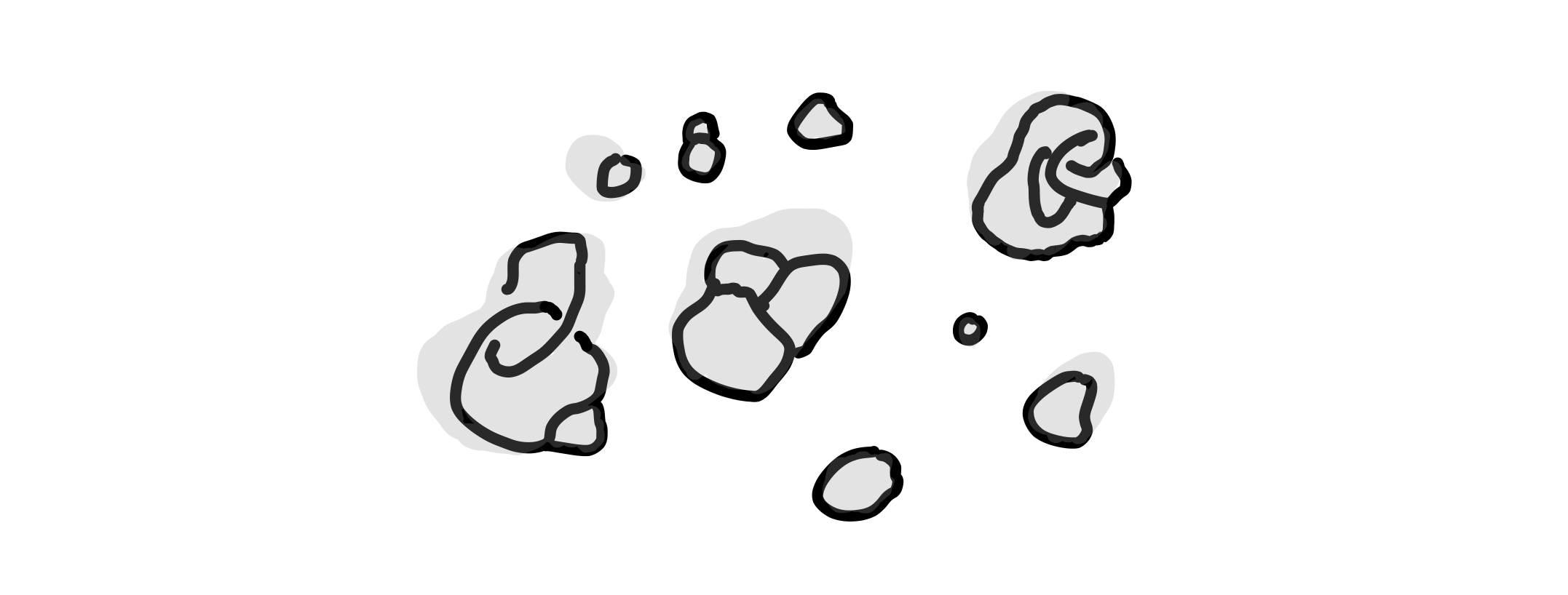}
\caption{The support of a correctable error (grey) and its string-like syndrome (black). The error is a product of correctable errors such that their syndromes are mutually disconnected. Each `elementary' correctable error is as local as possible, given the corresponding syndrome.}
\label{fig:typical_error}
\end{figure}

Naturally, characterizing the syndrome distribution is not enough to characterize a distribution of Pauli errors: errors with the same syndrome can differ in their logical effect. However, characterizing the syndrome is enough if we consider distributions of correctable Pauli errors%
\footnote{
Alternatively, an extra parameter may indicate the likelihood of uncorrectable errors~\cite{bombin:2015:single-shot}.
}.
For this kind of 3D codes, there exist a choice of correctable errors that is well suited to confined string-like syndromes~\cite{bombin:2015:single-shot}:
\begin{itemize}
\item
There exist a natural notion of connectedness for string-like syndromes, such that each connected component of a syndrome is a syndrome%
\footnote{
This is not always true for (the unlikely) syndromes of size comparable to the system~\cite{bombin:2015:single-shot}.
}.
\item
For connected syndromes the 'most local' class of errors is correctable%
\footnote{
Again, for syndromes of size comparable to the system there might be no clear choice.
}.
\item
A product of such `elementary' correctable errors with mutually disconnected syndromes is correctable.%
\end{itemize}
Figure~\ref{fig:typical_error} is a cartoon of a correctable error and its syndrome.

If the resulting family of correctable error distributions is to be useful, it has to be compatible with computation. Specifically,  if computation involves transversal  gates, their compatibility has to be ensured, \emph{i.e.} that the gates preserve the structure of the family. This problem is addressed in section~\ref{sec:propagation}%
\footnote{
Specifically for 3D color codes, but the results apply more broadly.
}.

\section{3D color codes}

This section reviews basic aspects of 3D color codes, making along the way a few new observations.

\subsection{3-Colexes}

\begin{warning}

A {\bf 3-colex} is a 3D lattice with~\cite{bombin:2007:branyons, bombin:2015:single-shot}
\begin{itemize}
\item
{\bf (3-)cells} labeled with 4 colors, 
\item
a (2D) boundary that is divided in {\bf facets} (connected sets of faces, called \emph{regions} in~\cite{bombin:2015:single-shot}) labeled with the same 4 colors,
\item
{\bf vertices} that belong to exactly a cell or facet of each color.
\end{itemize}
If two cells, or a cell and a facet, share a qubit, they meet at a single face. 

\end{warning}

\noindent Notice that no adjacent cells/facets have the same color. Each facet and cell forms a 2-colex, which is defined analogously, with 3-colored faces and (1D) boundaries. 

Given two different colors $\kappa_i$, $\kappa_j$, it is convenient to write 
\begin{equation}
\kappa_i\kappa_j
\end{equation} 
to represent the unordered pair of colors $\{\kappa_i, \kappa_j\}$. Edges and faces are attached colors as follows:
\begin{itemize}
\item
An edge $e$ has color $\kappa$ if the three cells/facets of which $e$ is part have colors different from $\kappa$.
\item
A face $f$ has label $\kappa_1\kappa_2$ if the two cells/facets of which $f$ is part have colors different from $\kappa_i$.
\end{itemize}

Below we only consider tetrahedral colexes, \emph{i.e.} 3-colexes with the topologogy of a tetrahedron, with each triangular facet labeled with a different color, see figure~\ref{fig:tetrahedral}. For an interesting family of tetrahedral colexes, see~\cite{bombin:2015:gauge}.

\begin{figure}
\includegraphics[width=\columnwidth]{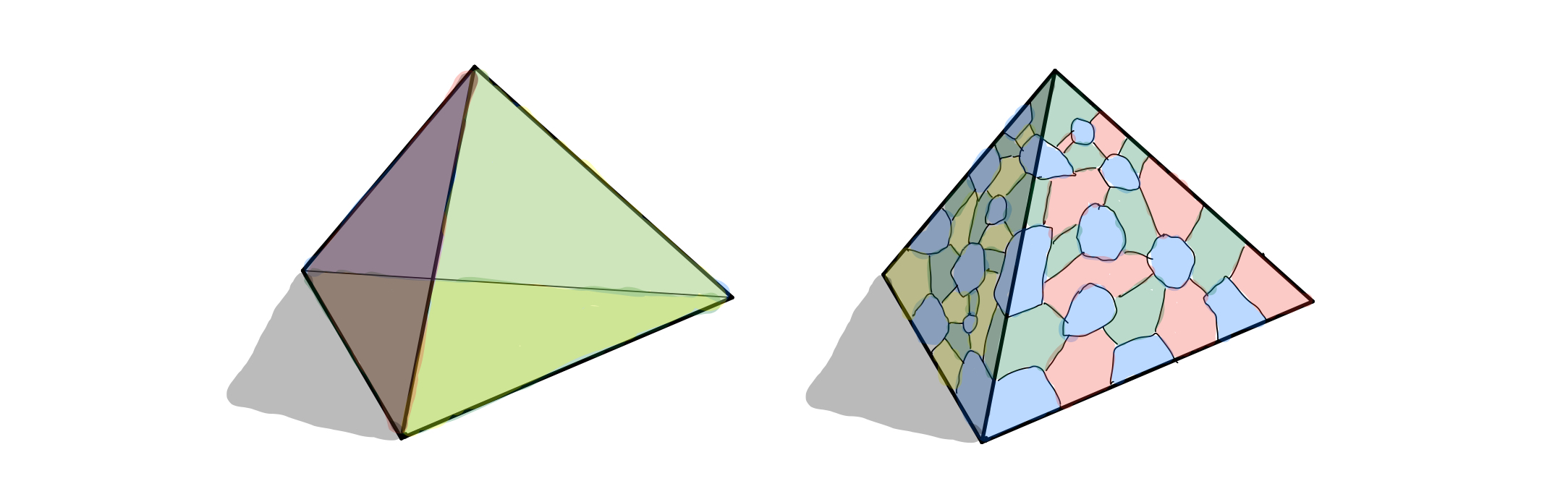}
\caption{Abstract topology (left) and lattice geometry (right) of a tetrahedral colex.}
\label{fig:tetrahedral}
\end{figure}

\subsection{Stabilizer group}

Color codes are CSS stabilizer codes that inherit the geometry of a colex. \emph{We identify cells, faces and facets with their sets of qubits}.

\begin{warning}

Given a  given a 3-colex, the corresponding {\bf 3D color code} is a stabilizer code with,
\begin{itemize}
\item
a qubit per vertex,
\item
a stabilizer generator $X_c$ per cell $c$,
\item
a stabilizer generator $Z_g$ per face $f$.
\end{itemize}

\end{warning}

\noindent That is, the stabilizer group is
\begin{equation}
S = S_XS_Z, \qquad S_X = \langle X_c\rangle, \qquad S_Z= \langle Z_f \rangle.
\end{equation}
where the indexation over the sets of cells and faces is implicit. Among the elements of $\mathcal Z(S)$ a central role is played by facet operators $X_r, Z_r$, with $r$ a facet~\cite{bombin:2015:single-shot}. 

\begin{warning}

{\bf Tetrahedral codes} are 3D color codes obtained from a tetrahedral 3-colex. 

\end{warning}
\noindent Tetrahedral codes have a single logical qubit. The logical $X$ and $Z$ operators can be chosen to be the $X_r$ and $Z_r$ facet operators for any of the triangular facets $r$~\cite{bombin:2015:single-shot}.

\subsection{Transversal gates}

Tetrahedral codes have a exceptional set of transversal gates: 

\begin{success}
\begin{itemize}[leftmargin=\baselineskip]
\item
The logical {\bf T gate} is performed by applying $T^{\pm b}$ to each physical qubit, for some odd $b$ and a certain sign pattern.
\item
As in any CSS code, the {\bf CNot gate} is transversal. It can be performed on any two codes with identical geometry, by applying CNot gates on the corresponding pairs of physical qubits.
\item
The controlled phase or {\bf CP gate} is transversal. It can be performed on any two codes with a pair of facets with identical geometry, by applying $CP$ gates on the corresponding pairs of qubits located on those facets.
\item
The {\bf P gate} can be performed on membrane-like subsets of physical qubits. In particular any facet suffices, as in the case of CP gates.
\end{itemize}
\end{success}

\noindent Appendix~\ref{sec:tetrahedral} provides the details. The transversal T and CNot gates were introduced in~\cite{bombin:2007:3dcc, bombin:2015:gauge}. It is worth pointing out that \emph{CNot gates are undesirable when locality matters}. If gates are local, the CNot gate requires the two codes to share the same space, unlike the $CP$ gate that only requires them to be side by side. 
The diagram of figure \ref{fig:gates} compares the four gates. 

\begin{figure}
\includegraphics[width=\columnwidth]{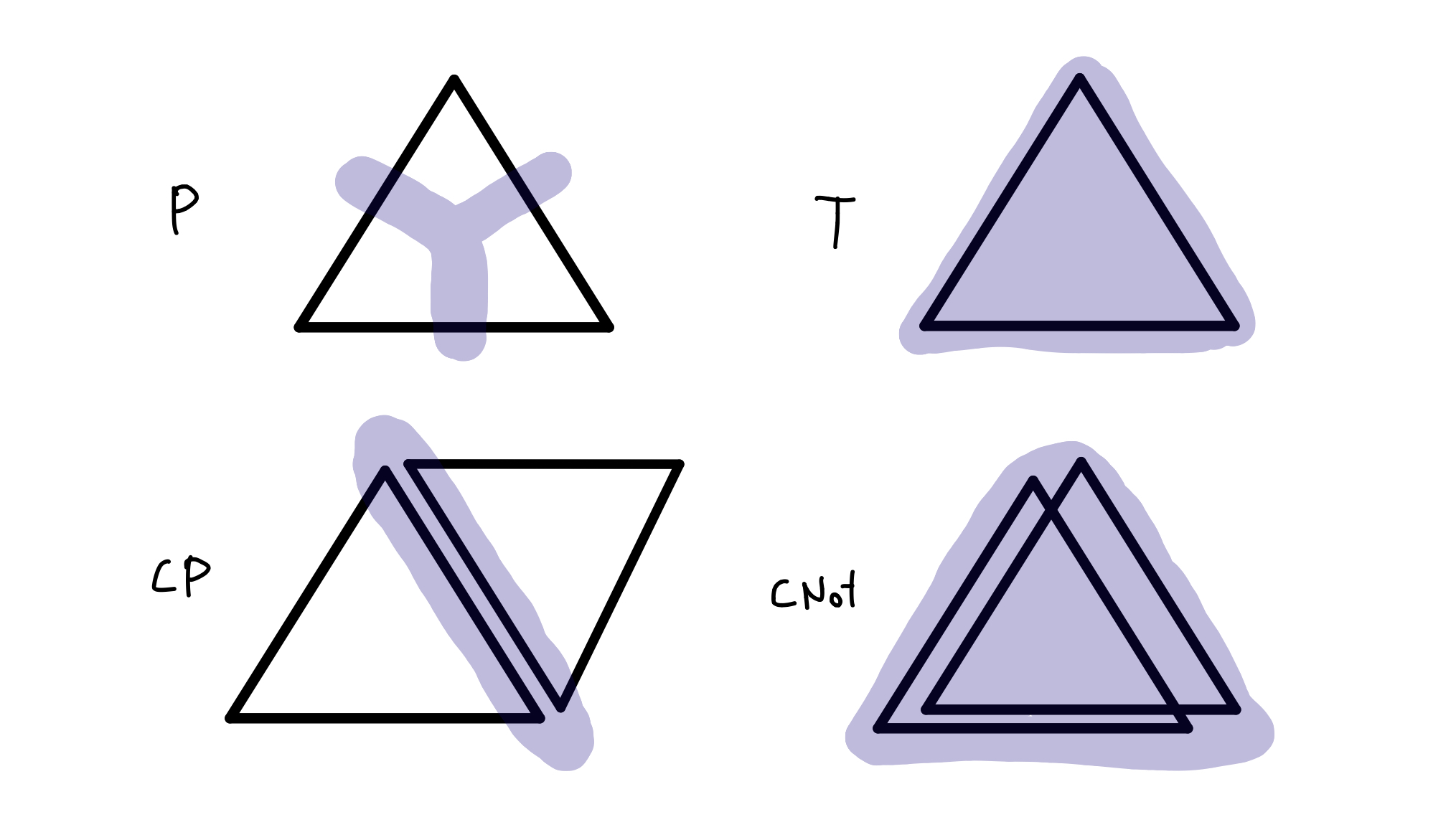}
\caption{The support of transversal gates in tetrahedral codes, with one spatial dimension removed for clarity.}
\label{fig:gates}
\end{figure}

\subsection{Error syndromes}

As in section \ref{sec:T:syndromes}, it is convenient to consider separately $S_X$ and $S_Z$ syndromes.

\subsubsection{Dual graph}

Syndromes are easiest to visualize in terms of the dual graph. This is essentially the 1-skeleton of the dual simplicial complex~\cite{bombin:2015:single-shot}, where faces become edges, with the peculiarity that facets give rise to edges with a single endpoint:

\begin{warning}

The {\bf dual graph} of a 3-colex has
\begin{itemize}
\item
an edge per face of the colex, and
\item
a vertex per cell of the colex.
\end{itemize}
An edge dual to a face $f$ has endpoints on the vertices dual to the cells of which $f$ is a face.

\end{warning}

\noindent In the dual picture vertices are 4-colored (with a color $\kappa_i$, as their dual cells) and edges are 6-colored (with a label $\kappa_i\kappa_j$, as their dual faces). In particular, \emph{the colors of a dual edge are complementary to the colors of its endpoints}, see figure~\ref{fig:dual_graph}.

\begin{figure}
\includegraphics[width=\columnwidth]{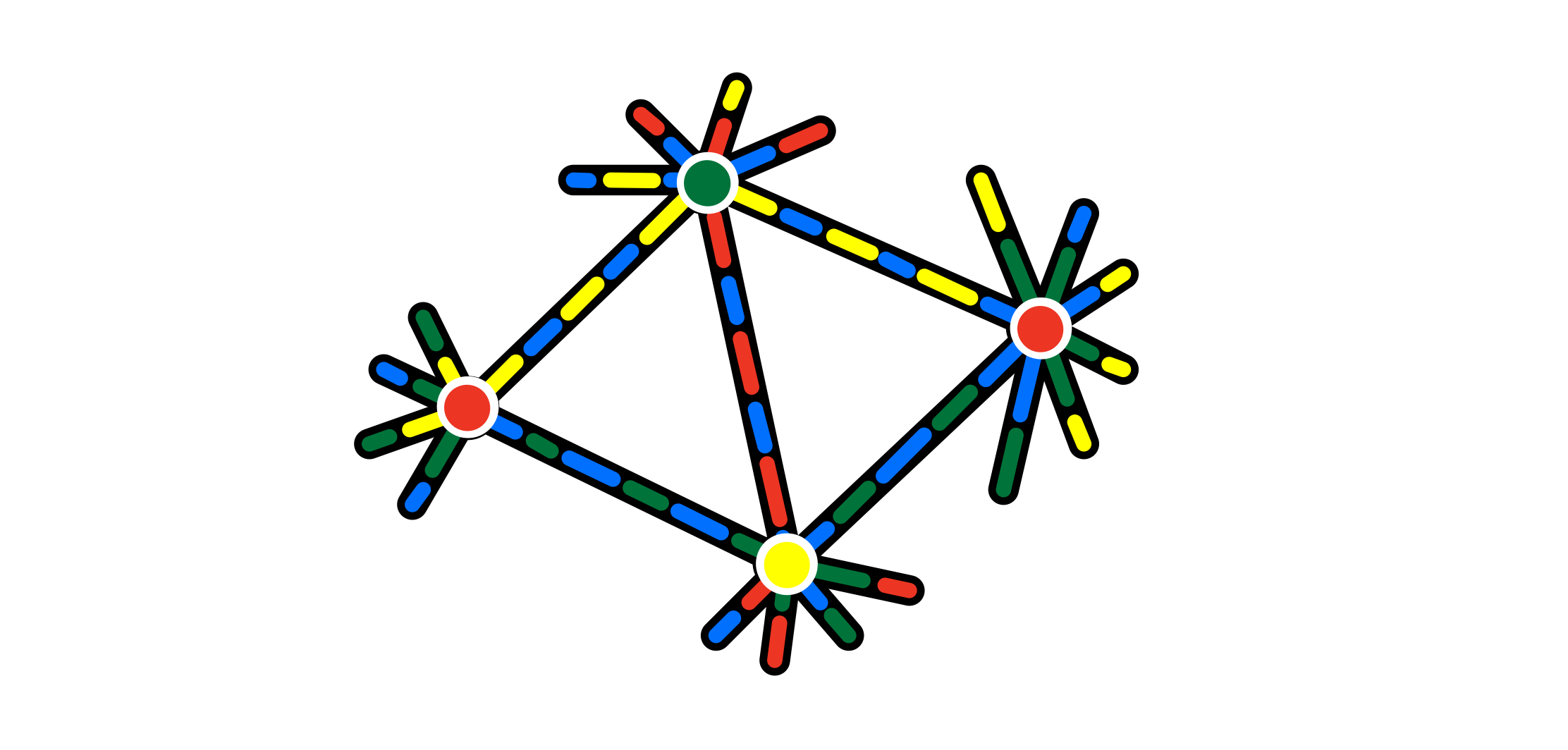}
\caption{The dual graph has 4-colored vertices. Edges are labeled with the two colors complementary to the colors of their endpoints.}
\label{fig:dual_graph}
\end{figure}

\subsubsection{Electric charge}

The generators of $S_X$ are cell operators $X_c$. We identify each syndrome $\sigma$ of $S_X$ with the subset $\xi$ of cells $c$ with $\sigma(c)=-1$, and regard $\xi$ as a $\mathbf Z_2^3$ 'electric' charge configuration. In the dual graph picture $\xi$ is a set of vertices, and each vertex carries a charge corresponding to its color. 
\begin{warning}
The $\mathbf Z_2^3$ {\bf charge group} is generated by the colors $\kappa_i$ via the presentation
\begin{equation}
\kappa_1 + \kappa_1 = \kappa_1+\kappa_2+\kappa_3+\kappa_4 = 0,
\end{equation}
where the $\kappa_i$ are arbitrary but all different. 
\end{warning}
This definition is meaningful because charge is conserved. Indeed, given a colex and a color $\kappa$, let $W(\kappa)$ denote the set of $\kappa$-colored facets and cells. The conservation of electric charge is manifest in the following relations~\cite{bombin:2015:single-shot}
\begin{equation}
\prod_{w\in W(\kappa)} X_w = \prod_{w\in W(\kappa')} X_w.
\label{charge_conservation}
\end{equation}
For tetrahedral codes there are no further relations and thus any subset of cells is a syndrome.

\subsubsection{Magnetic flux}

The generators of $S_Z$ are face operators $Z_f$. We identify a syndrome $\sigma$ of $S_Z$ with the subset $\phi$ of faces $f$ with $\sigma(f)=-1$,  and regard $\phi$ as a $\mathbf Z_2^3$ 'magnetic' flux configuration. In the dual graph picture $\phi$ is a set of edges, each carrying a flux unit corresponding to its label. 
\begin{warning}
The $\mathbf Z_2^3$ {\bf flux group} is generated by the color pairs $\kappa_i\kappa_j$, via the presentation
\begin{equation}
\kappa_1\kappa_2+\kappa_1\kappa_2=\kappa_1\kappa_2+\kappa_2\kappa_3+\kappa_3\kappa_1=0,
\end{equation}
where the $\kappa_i$ are arbitrary but all different. 
\end{warning}
This definition is meaningful because flux is conserved. This becomes apparent by attaching a \emph{monopole configuration} $\partial \phi$ to each flux configuration $\phi$: a mapping from the set of vertices to the flux group taking each vertex $v$ to the sum of the labels of the edges of $\phi$ that are incident on $v$. Importantly, $\partial\phi$ dictates whether $\phi$ is the syndrome of some element of $P_X$~\cite{bombin:2015:single-shot}:

\begin{success}

In a tetrahedral code, a set of dual edges $\phi$ is a syndrome iff it satisfies {\bf Gauss's law}, i.e.
\begin{equation}
\partial \phi = 0.
\end{equation}
\end{success}

\noindent It is worth noting the folloing points: 
\begin{itemize}
\item
Syndromes are characterized by the continuity of color lines (any given color $\kappa$ forms closed loops), see figure~\ref{fig:color_conservation}.

\begin{figure}
\includegraphics[width=\columnwidth]{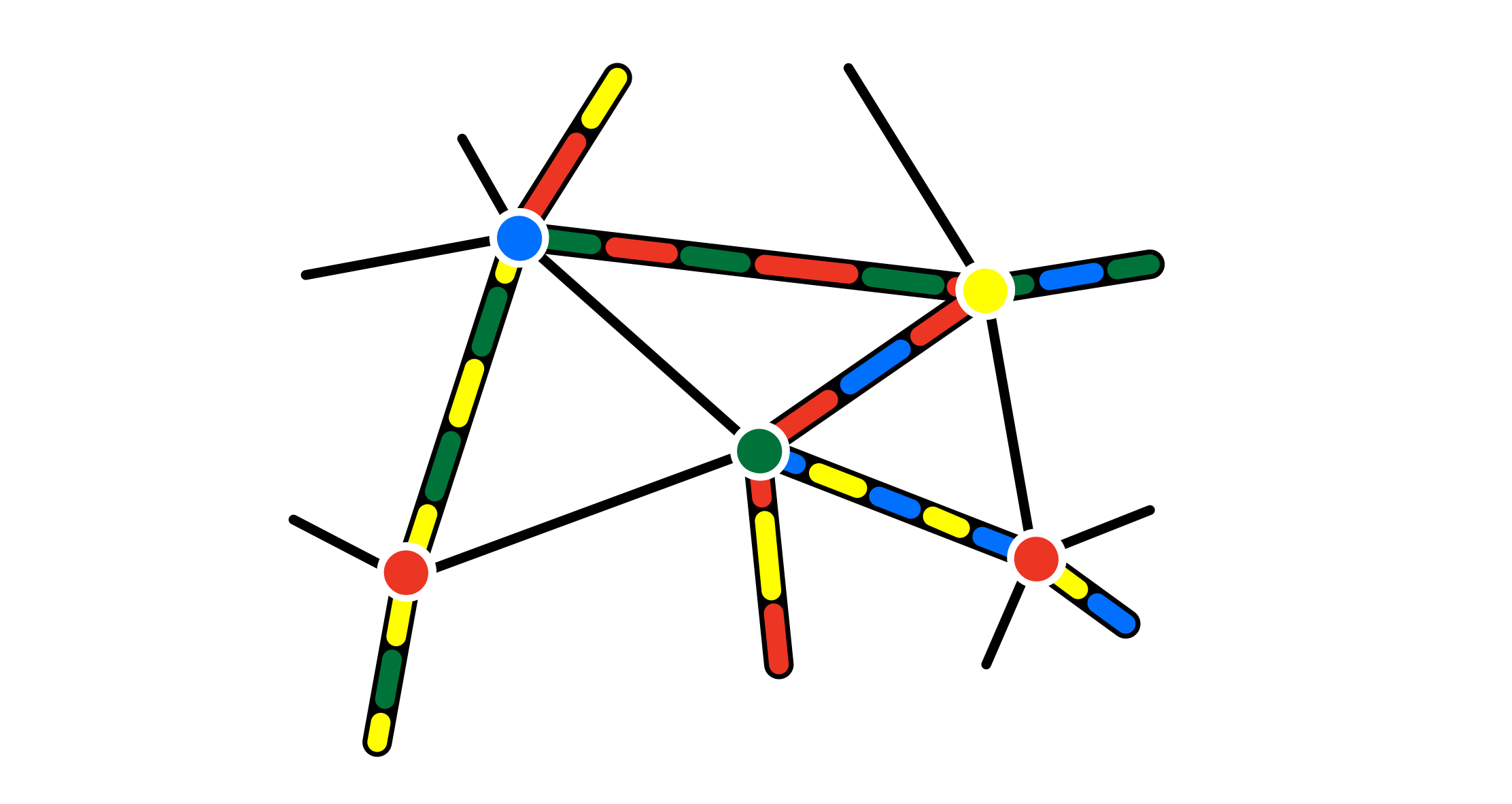}
\caption{In a flux configuration satisfying Gauss's law, any given color forms closed loops.}
\label{fig:color_conservation}
\end{figure}

\item
Electric charges and magnetic monopoles have a very different nature. In the Hilbert space where the 3D color code lives the only possible magnetic flux configurations $\phi$ are syndromes. Flux conservation is kinematic and thus there are no monopoles. In contrast, the absence of electric charges is a property of encoded states, not of the Hilbert space.
\item
When $\phi$ represents noisy outcomes of $Z_f$ measurements, $\partial \phi$ is a 'syndrome of a syndrome', a key concept for single-shot error correction~\cite{bombin:2015:single-shot}. In particular, if $\partial p$ is the syndrome (in $S_Z$) of some $p\in P_X$ then
\begin{equation}
\partial^2 p = 0.
\end{equation}
\item
Any set of faces (or dual edges) $\phi$ can regarded as a binary vector by identifying $\phi$ with its indicator function. In this picture syndromes form a classical code, with two check operators per cell (or inner dual vertex).
\item
The asymmetry between $S_X$ and $S_Z$ has important implications for fault-tolerant error correction. Since $X$-error syndromes, i.e. syndromes of $S_Z$, take the form of connected string-like objects, single-shot error correction tools apply and local $X$ errors can be dropped in favor of correctable $X$ errors with local syndromes. The same is not true for $Z$ errors%
\footnote{
We are \emph{not} considering gauge color code techniques here~\cite{bombin:2015:gauge}.
}.
\end{itemize}
 
\subsection{Topological interaction}

From a condensed matter perspective the 3D color code subspace is the ground state subspace of a topologically ordered system~\cite{bombin:2007:branyons}. In particular, it can be regarded as a string-net condensate or, dually, a membrane-net condensate. In this picture the charge and flux excitations are subject to a topological interaction. This section reviews the operator-related aspects of the physics.

\subsubsection{String operators}

\begin{warning}

The {\bf string operator} for a given color $\kappa$ and a path $s$  on a 3-colex (with no duplicated edges) is
\begin{equation}
Z_{s,\kappa} := \prod_{e\in E(s,\kappa)} Z_e,
\end{equation}
where $E(s,\kappa)$ is the set of edges of $s$ with color $\kappa$.

\end{warning}

\noindent Regarded as an error, $X_{s,\kappa}$ has a syndrome
\begin{equation}
\partial Z_{s,\kappa}=\{c_1\}+\{c_2\},
\end{equation}
where $c_1$ (respectively $c_2$) is the only $\kappa$-cell adjacent to the first (last) vertex of $s$. The string operator transfers a charge $\kappa$ between its endpoints, i.e. it creates a $\kappa$ unit of electric field along its path. In particular, for closed paths $s$ the string operators have trivial syndrome. Moreover, when in addition $s$ has trivial $\mathbf Z_2$ homology its string operators belong to the stabilizer. This shows that the code subspace can be regarded as a string condensate.

\subsubsection{Membrane operators}

\begin{warning}

The {\bf membrane operator} for a given color pair $\kappa_1\kappa_2$ and a surface $m$ on a 3-colex (a 2-manifold-like set of faces) is
\begin{equation}
X_{m,\kappa_1\kappa_2} := \prod_{f\in F(m,\kappa_1\kappa_2)} X_m,
\end{equation}
where $F(m,\kappa_1\kappa_2)$ is the set of faces of $s$ labeled with the color set complementary to $\kappa_1\kappa_2$.

\end{warning}

\noindent Regarded as an error, $X_{s,\kappa}$ has a syndrome
\begin{equation}
\partial X_{m,\kappa_1\kappa_2}=\sum_v \{f_v\},
\end{equation}
where $v$ is the list of vertices along the boundary of $m$, and $f_v$ is the only $\kappa_1\kappa_2$ face adjacent to $v$. The membrane operator creates a unit of $\kappa_1\kappa_2$ magnetic flux along its boundary. In particular, for closed surfaces $m$ the membrane operators have trivial syndrome. Moreover, when in addition $m$ has trivial $\mathbf Z_2$ homology its membrane operators belong to the stabilizer. This shows that the code subspace can be regarded as a membrane condensate.

\subsubsection{Duality}\label{sec:duality}

Given a path $s$ and a surface $m$ that intersect once (in the sense that $s$ approaches and leaves the surface $m$ on different sides), the corresponding string and membrane operators have commutation rules that only depend upon the charge and flux they create. Namely,
\begin{equation}
(Z_{\kappa_1}, X_{m,\kappa_2\kappa_3}) = \langle\kappa_1,\kappa_2\kappa_3\rangle := 
\begin{cases}
-1, &\text{if $\kappa_1\in\kappa_2\kappa_3$},\\
1, &\text{otherwise}.
\end{cases}
\end{equation}
From a condensed matter perspective, this gives rise to topological interactions between the corresponding charge and flux excitations~\cite{bombin:2007:branyons}. 

\begin{figure}
\includegraphics[width=\columnwidth]{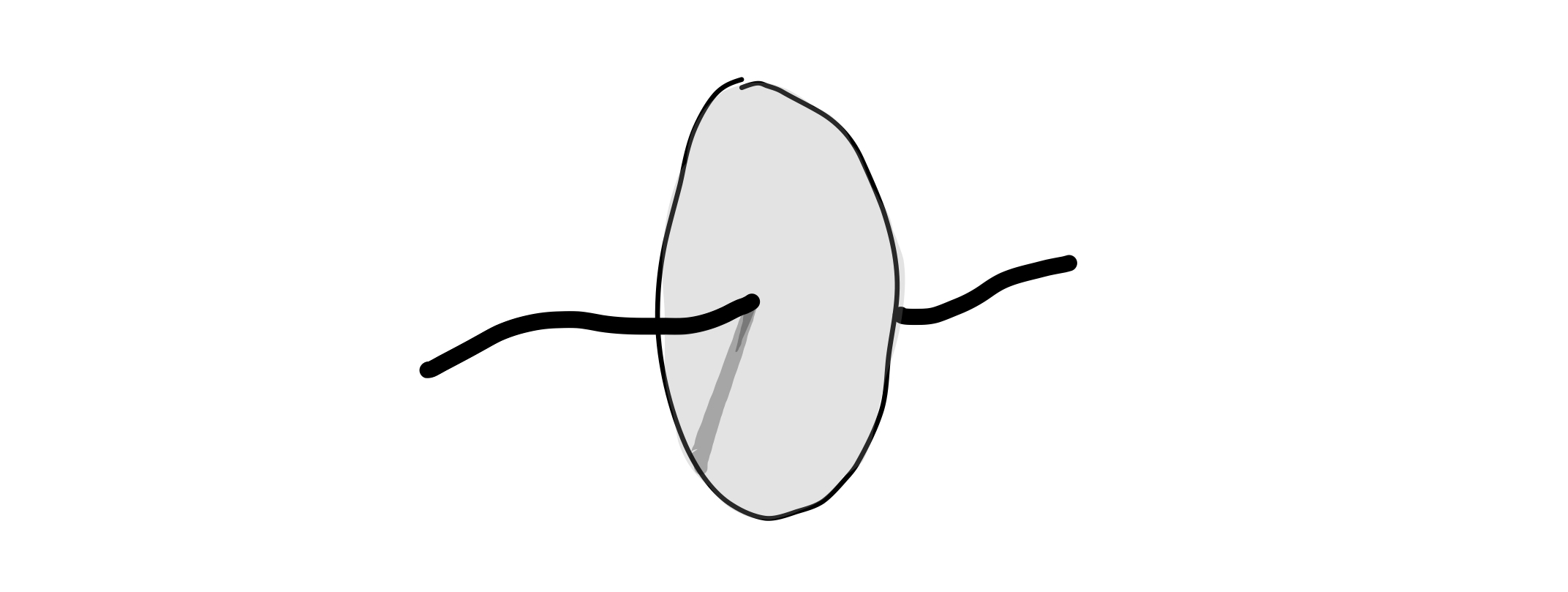}
\caption{A membrane pierced by a string. The corresponding membrane and string operators are subject to commutation rules that depend only on their charge and flux labels.}
\label{fig:pierced_membrane}
\end{figure}

Each color $\kappa$ yields an irrep $\langle \kappa,\cdot \rangle$ of the flux group. These irreps generate the full group of irreps of the flux group, providing a natural duality between the charge and flux groups. This has a useful application. Suppose that some string-like operator $s$ is given,~\emph{i.e.} an operator with support along a string-like region and that commutes with all check operators except those in its endpoint regions. The charge transferred by $s$ can be identified from its commutation relations with the different membrane operators of a surface pierced by its string, as in figure~\ref{fig:pierced_membrane}. And vice versa, the same trick can be applied to a membrane-like operator.

\newcommand{\qset}[1]{Q({#1})}

\section{Error propagation}\label{sec:propagation}

This section studies error propagation for the transversal gates of tetrahedral color codes. The focus is on distributions of correctable $X$ Pauli errors with local syndromes, in the sense of section~\ref{sec:confined}. The technical results of this section can be found in appendices~\ref{sec:noise_Clifford} and~\ref{sec:noise_T}.

Section~\ref{sec:linking} farther illuminates the physics behind the T gate error propagation described here.

\subsection{Correctable errors}

Working with syndrome distributions requires choosing a set of correctable errors. As discussed in section~\ref{sec:confined}, a convenient approach is to separate each syndrome into its connected components and to attach an error separately to each of them. This is possible as long as the connected component of a syndrome is itself a syndrome, as is the case for $X$-error syndromes in tetrahedral codes.

The criteria to attach a correctable error to a connected syndrome, as per the recipe of~\ref{sec:confined}, is to choose the `most local error'. In tetrahedral codes the following result provides a straightforward choice:

\begin{lem}\label{lem:globality}
If an $X$-error syndrome $\phi$ has no faces on a facet $r$ there exist $x\in P_X$ with syndrome $\phi$ and no support in $r$.
\end{lem}
\noindent
For any $\phi$ that does not connect the four facets of the tetrahedron this result defines a correctable class of errors%
\footnote{
As shown in section~\ref{sec:tolerable_cc} an element of $P_X$ with no support on a facet is tolerable, which shows that this is a consistent definition.
}.
This is enough to define, as per the recipe, correctable errors for all syndromes such that none of their connected components has faces in all four facets. These are the \emph{typical} syndromes when operating in a confined regime: connected components with a length comparable to the system size are unlikely.

\subsubsection{Tolerable errors}\label{sec:tolerable_cc}

Recall from section~\ref{sec:T} that the tolerable errors induced by the transversal T gate form a set of correctable errors. Interestingly, tolerable errors form a superset of the correctable errors defined by the connected-component-based approach discussed here. Indeed,
\begin{itemize}
\item
if $X_\alpha$ has no support on a given facet $r$, then we can choose $X_\lambda = X_r$ in lemma~\ref{lem:tolerability}(iii), showing that $X_\alpha$ is tolerable, and 
\item
the product of tolerable errors with mutually disconnected syndromes is tolerable, by the following result and lemma~\ref{lem:factorization}(i).
\end{itemize}

\begin{lem}\label{lem:separability_cc}
If the $X$-error syndromes $\phi_i$ are mutually disconnected, then they are separated.
\end{lem}

\noindent
The recipe of section~\ref{sec:confined} to choose correctable errors was originally introduced ad-hoc to prove the existence of a threshold for single-shot error correction~\cite{bombin:2015:single-shot}. Remarkably, it coincides (were defined) with a set correctable errors naturally defined by an a priori unrelated transversal gate.

Is it possible to apply these ideas to practical error correction? Deciding if an $X$ error is tolerable requires deciding whether a system of binary linear equations has a solution, with the problem size growing like the number of physical qubits. Thus it is unlikely that there is any gain when compared with the straightforward approach based on separating the syndrome in connected components. And leaving computational complexity aside, it is not even clear that this particular way of choosing correctable errors is sensible, say, for small codes.

\subsection{Transversal Clifford gates }

This section briefly discusses the propagation of errors for logical Clifford gates. The emphasis is on the propagation of the syndrome: since errors do not propagate beyond their support, this is enough to characterize the propagated error as long as the source error is localized.

\subsubsection{Standard P gate}

\begin{figure}
\includegraphics[width=\columnwidth]{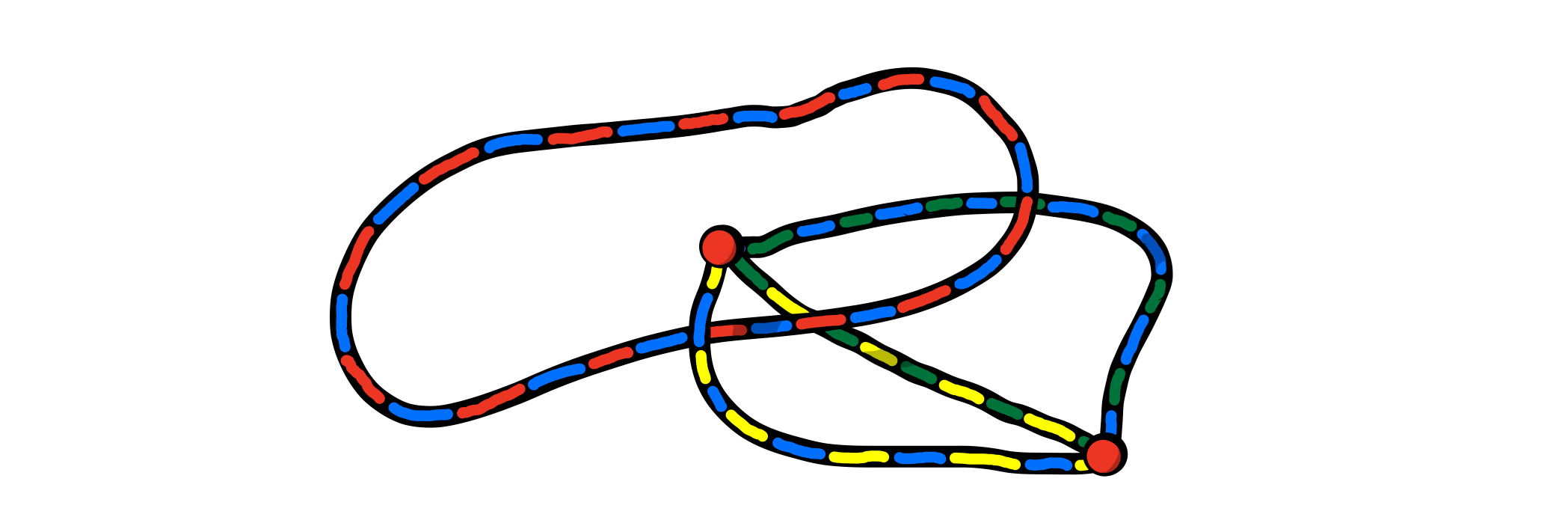}
\caption{A flux configuration and its branching points (colored vertices).}
\label{fig:branching}
\end{figure}

\begin{warning}

$\text{br}(\phi)$ is the set of {\bf branching points} of the flux configuration $\phi$: cells where $\phi$ has an odd number of faces of any given color label.

\end{warning}

\noindent The P gate can be implemented transversally by applying twice the transversal T gate. It commutes with $Z$ errors. As shown in appendix~\ref{sec:noise_Clifford} the propagation of $X$ errors takes the form:
\begin{equation}
X_\alpha \mapsto X_\alpha Z_\beta,\qquad \partial Z_\beta = \text{br}(\partial X_\alpha).
\end{equation}
The relationship between the syndrome $\phi$ and the charges $\text{br}(\phi)$ is depicted in figure~\ref{fig:branching}.

The same relationship between flux and charges occurs in a different scenario: single-shot error correction in gauge color codes~\cite{bombin:2015:single-shot}. In that scenario the flux is the wrongly reconstructed part of the 'gauge syndrome', and the charges are the syndromes of the residual noise. Otherwise the situation is equivalent: if the flux is sufficiently confined the propagated $Z$ errors follow a local distribution.

\subsubsection{Facet P gate}

\begin{figure}
\includegraphics[width=\columnwidth]{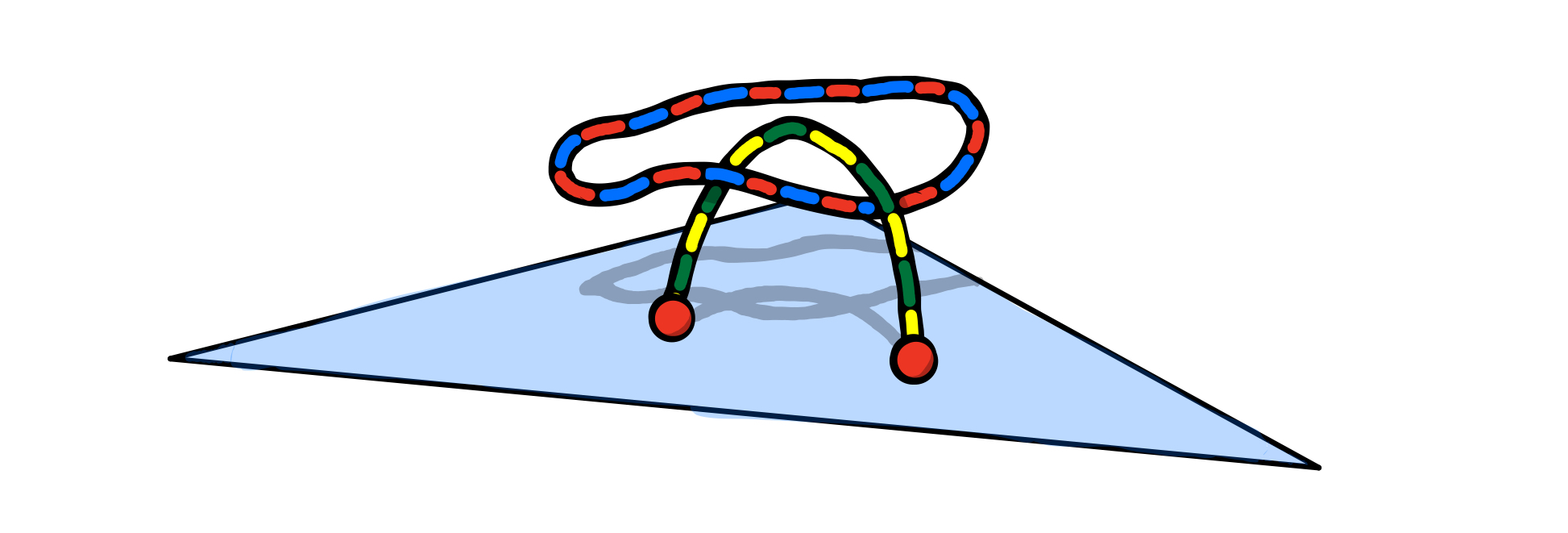}
\caption{A flux configuration and its endpoints on the flue facet.}
\label{fig:endpoints_facet}
\end{figure}

\begin{warning}
$\text{end}_r(\phi)$ is the set of {\bf endpoints} of the flux configuration $\phi$ on the facet $r$: cells sharing a face with $r$ that belongs to $\phi$.
\end{warning}

\noindent 
The P gate can be implemented within a much more restricted support: the support of any logical $X$ operator (which are membrane-like), see appendix~\ref{sec:tetrahedral}. In particular, the P gate can be implemented on a single facet $r$ of the tetrahedron. This gate also commutes with $Z$ errors. As shown in appendix~\ref{sec:noise_Clifford} the propagation of $X$ errors takes the form:
\begin{equation}
X_\alpha \mapsto X_\alpha Z_\beta,\qquad \partial Z_\beta = \text{end}_r(\partial X_\alpha).
\end{equation}
It is worth noting that $Z_\beta$ can always be chosen to have its support contained in the facet $r$. Figure~\ref{fig:endpoints_facet} depicts the relationship between the syndrome $\phi$ and the charges $\text{end}_r(\phi)$.

These relationship between fluxes and charges occurs also in a known scenario: dimensional jumps~\cite{bombin:2016:dimensional}. A dimensional jump involves a form of single-shot error correction in which the flux is the wrongly reconstructed part of an error syndrome, and the charges are the syndromes of the corresponding residual noise (on the 2D color code of the facet $r$). As in that scenario, if the flux is sufficiently confined the propagated $Z$ errors follow a local distribution.

\subsubsection{Facet-to-facet CP gate}

The facet-to-facet CP gate is very closely related to the facet $P$ gate. The only difference is that $Z$ errors are propagated to the other code's facet, rather than to the facet where the original $X$ error lives.

\subsubsection{CNot gate}

The CNot gate propagates (copies) $X$ errors from the source code to the target code, and conversely for $Z$ errors. The resulting distribution of errors can be obtained by composing two such distributions, and this is compatible with confined syndromes~\cite{bombin:2015:single-shot}.

\subsection{Transversal T gate}\label{sec:propagation_T}

This sections discusses the propagation of Pauli error under the transversal $T$ gate. This gate commutes with $Z$ errors, so that only $X$ error propagation requires attention.
Given the definitions and results of section~\ref{sec:T}, the main aim here is to understand, for any given flux configuration $\phi$, the erasure of information due to $H(\phi)$ and the support of the errors in $E(\phi)$, particularly in connection with local syndrome distributions. The proofs are in appendix~\ref{sec:noise_T}.

\subsubsection{Check operator erasure}

As discussed in section~\ref{sec:erasure}, there exist scenarios in which only the syndrome of $\mathcal Z_S(H(\phi))$ is relevant for error correction. For tetrahedral codes there is a subgroup of $\mathcal Z_S(H(\phi))$ that has a simple description and physical interpretation, and that should give similar results in practice.
To describe it, the first step is to factorize $H(\phi)$ using lemma~\ref{lem:separability_cc} and the following result.

\begin{lem}\label{lem:outside}
\label{lem:Zerror}
If $w$ is a cell or facet with no faces in the $X$-error syndrome $\phi$, then $X_w$ commutes with the elements of $E(\phi)$.
\end{lem}
\noindent
By lemma~\ref{lem:separability_cc}, the group $H(\phi)$ takes the form
\begin{equation}
H(\phi)=\prod_i H(\phi_i),
\end{equation}
where $\phi_i$ are the connected components of $\phi$. And by lemma~\ref{lem:outside}, the elements of $H(\phi_i)$ commute with the operators $X_w$ of all cells and facets $w$  with no faces in $\phi_i$.
Thus, the following two kind of check operators generate a subgroup of $\mathcal Z_S(H(\phi))$:
\begin{itemize}
\item
cell operators from cells with no faces in $\phi$.
\item
for each $\phi_i$ and each color combination $\kappa\kappa'$ such that $\phi_i$ has no endpoints in the $\kappa$ and $\kappa'$ facets, the product of all cell operators from $\kappa$ and $\kappa'$ cells with faces of in $\phi_i$.
\end{itemize}
The latter operators commute with the elements of $H(\phi)$ by virtue of the charge conservation relations \eqref{charge_conservation}. This is an expression of the invariance, under $H(\phi)$, of the total charge in the cells connected by $\phi_i$ up to the charge that they can exchange with the facets where $\phi_i$ has endpoints. A clear physical picture emerges: under the action of $H(\phi)$ cells and facets only exchange charge if they belong to the same connected component of $\phi$.

\begin{success}
\center
$H(\phi)$ transports charge only along the syndrome $\phi$.
\end{success}

\subsubsection{Error support}\label{sec:support}

The aim here is to (i) reveal the close geometrical relationship between an $X$-error syndrome $\phi$ and the support of the elements of $E(\phi)$, and (ii) study the locality of propagated errors for local syndrome distributions, via~\ref{thm:T}.

\begin{warning}
$\qset{\phi}$ is the {\bf qubit set} of \emph{some} connected subgraph of the colex such that every cell and facet with a face in $\phi$ contains also at least a qubit of $\qset{\phi}$.
\end{warning}

\begin{figure}
\includegraphics[width=\columnwidth]{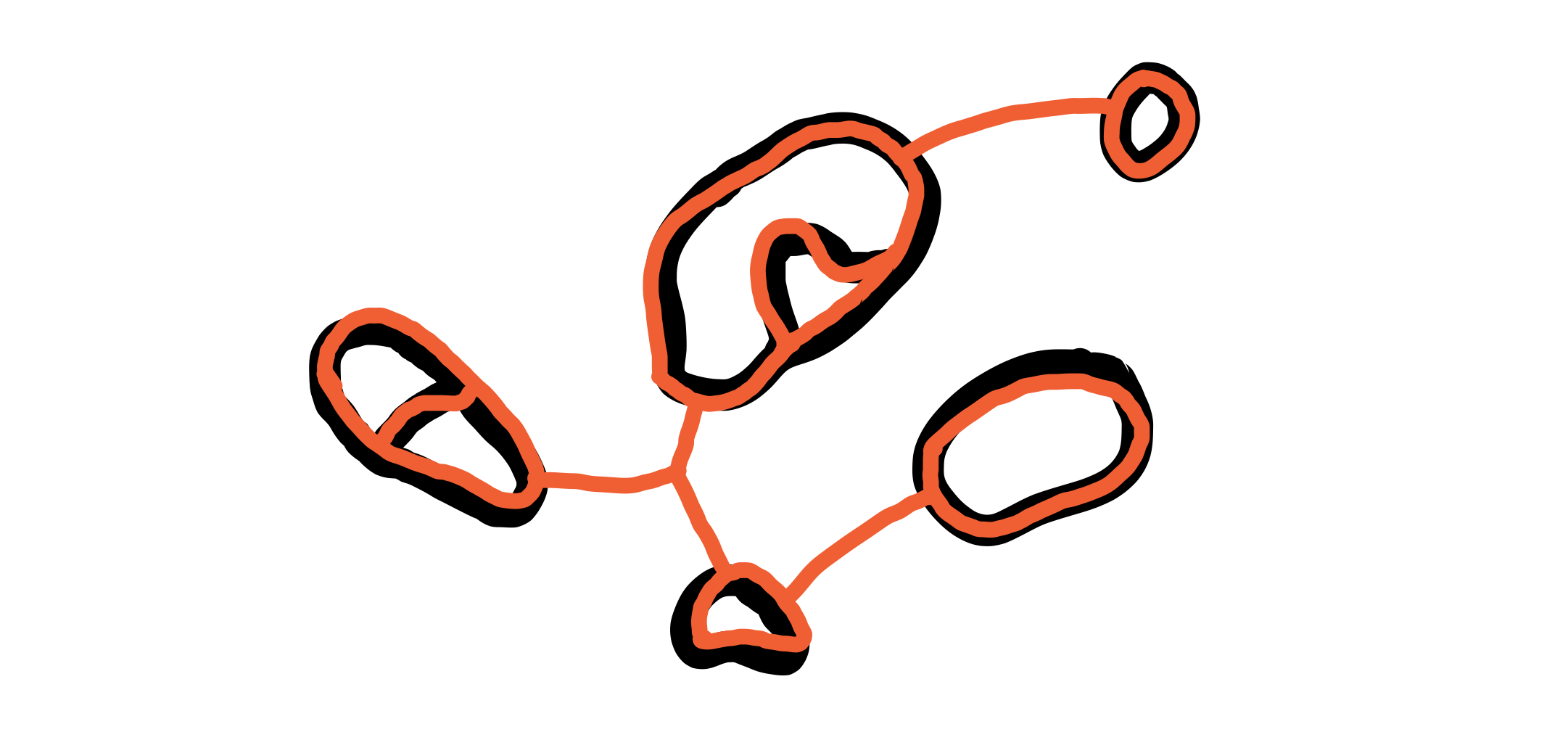}
\caption{A syndrome composed of several connected components (black) and a possible corresponding connected subgraph of the colex (orange).}
\label{fig:Qphi}
\end{figure}

\noindent
This construction is illustrated in figure~\ref{fig:Qphi}. The sets $\qset{\phi}$ are arbitrary within the given constraint. When $\phi$ is connected $\qset{\phi}$ can be a subset of of the qubits of cells with faces in $\phi$.

\begin{lem}\label{lem:Zerror}
For any $X$-error syndrome $\phi$, every element of $E(\phi)$ is equivalent, up to a stabilizer, to an operator with support contained in $\qset{\phi}$.
\end{lem}

\noindent Consider again the connected components $\phi_i$ of $\phi$. By lemmas~\ref{lem:factorization} and~\ref{lem:Zerror}, every element of $E(\phi)$ is the product of
\begin{itemize}
\item
for each $i$, an element of $E(\phi_i)$ with support contained in $\qset{\phi_i}$, and
\item
for each pair $i, j$, some $z(\phi_i,\phi_j)\in P_Z$ with support contained in $\qset{\phi_i+\phi_j}$ and, up to an element of $H(\phi_i+\phi_j)$,  of the form
\begin{equation}\label{eq:zij}
Z_{\alpha_i\cap\alpha_j}
\end{equation} 
for any tolerable $X_{\alpha_i}$, $X_{\alpha_j}$ with syndromes $\phi_i$, $\phi_j$, respectively.
\end{itemize}
Let $z(\phi_i,\phi_j)$ be trivial if it belongs to $H(\phi_i+\phi_j)$.
The connected components of $\phi$ can be arranged in clusters: if $z(\phi_i,\phi_j)$ is non-trivial, then $\phi_i$ and $\phi_j$ belong to the same cluster. This partition of the syndromes $\phi_i$ yields \begin{equation}
E(\phi) = \prod_k E(\phi[k]),\qquad \phi[k]:=\sum_{j\in I(k)} \phi_j,
\end{equation}
where $I(k)$ is the $k$-th cluster's set of indices. That is, up to stabilizers every element of $E(\phi)$ has support contained in the union of the sets $\qset{\phi[k]}$.

An assumption about the colex is required to make further progress.

\begin{warning}
If $z(\phi_i,\phi_j)$ is non-trivial there exists a path connecting $\qset{\phi_i}$ to $\qset{\phi_j}$ with length bounded by
\begin{equation}
k_0\min (|\phi_i|, |\phi_j|).
\end{equation}
\end{warning}

\noindent
This assumption is justified by a result of section~\ref{sec:linking}: if $z(\phi_i,\phi_j)$ is nontrivial $\phi_i$ and $\phi_j$ are (topologically) linked. To enforce it for a family of tetrahedral codes for a fixed constant $k_0$, all the codes should share a single uniform local structure, including at facets and corners, as \emph{e.g.} the family described in~\cite{bombin:2015:gauge}.
It is an easy exercise to show that, under this assumption, for each cluster the connected graphs can be chosen so that
\begin{equation}
|\qset{\phi[k]}-\bigcup_{i\in I(i)} \qset{\phi_i}|\leq 
k_0\phi[k].
\end{equation}
A standard argument based on counting connected structures~\cite{dennis:2002:tqm, bombin:2015:single-shot} shows that a local distribution of syndromes $\phi$ yields a local distribution of errors in $P_Z$ if elements of $E(\phi)$ are chosen at random (or otherwise) for each $\phi$.

\section{Linking charge}\label{sec:linking}

As discussed in section~\ref{sec:support}, under a logical $T$ gate an $X$ error gives rise not only to propagated noise in the vicinity of its syndrome, but also to the terms $z(\phi_1,\phi_2)$. This section clarifies their physical meaning.

\subsection{Charge transfer}

\begin{figure}
\includegraphics[width=\columnwidth]{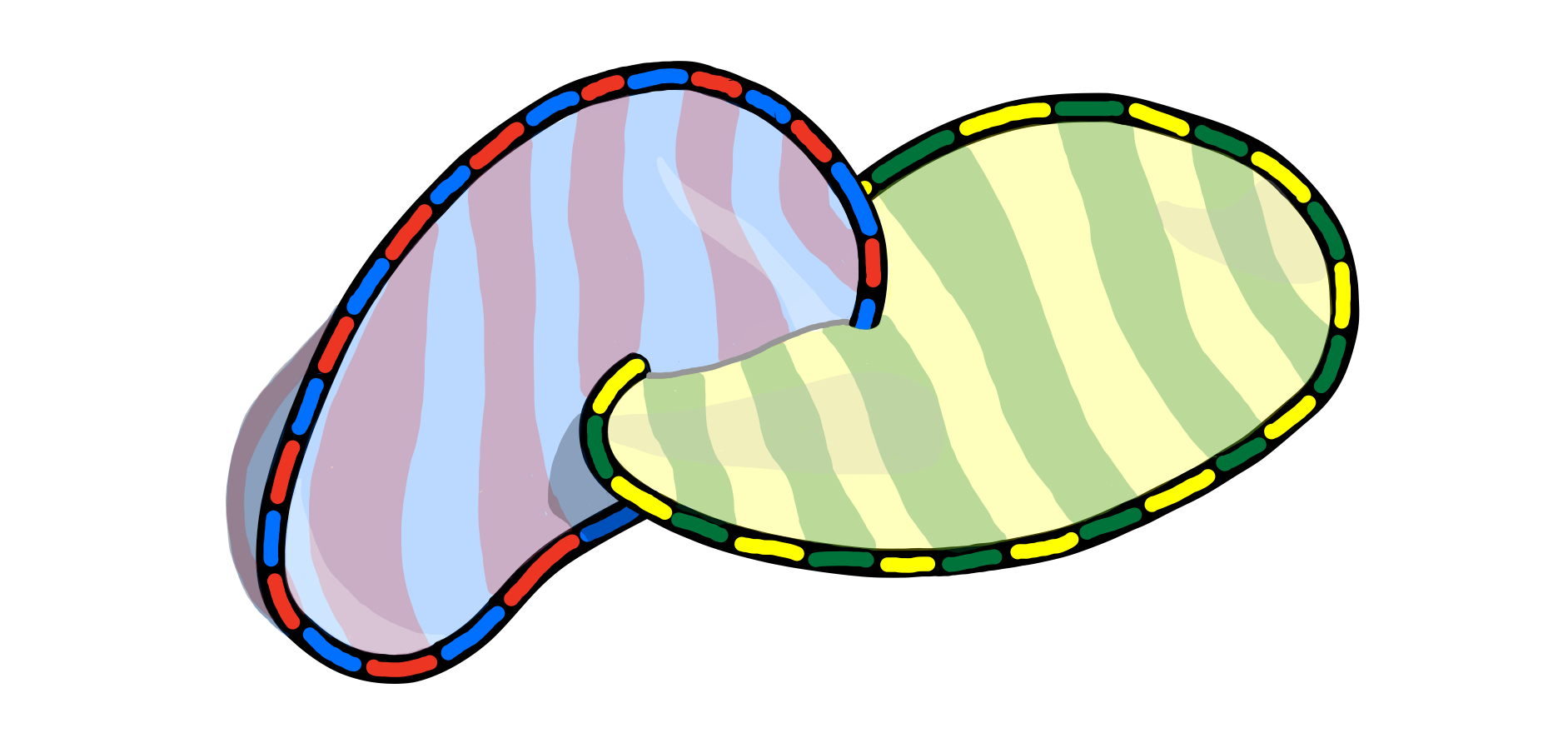}
\caption{A flux configuration consistent of two loops with linking number one, and the corresponding membranes. A grey line indicates the region where the membranes intersect.}
\label{fig:link}
\end{figure}

Let $\phi_1$ and $\phi_2$ be two loops carrying flux $\kappa_{11}\kappa_{12}$ and $\kappa_{21}\kappa_{22}$, respectively, and with linking number one, as depicted in figure~\ref{fig:link}.
Consider two disc-shaped surfaces $m_i$, $i=1,2$, such that
\begin{equation}
\partial X_{m_i,\kappa_{i1}\kappa_{i2}} = \phi_i.
\end{equation}
Specifically, the surfaces are such that they intersect only along a string-like region connecting the two loops, rather than in more complicated ways, see figure~\ref{fig:link}. 
Up to an element of $H(\phi_1+\phi_2)$ the operator $z(\phi_1,\phi_2)$ takes the form 
\begin{equation}
Z_{\alpha_1\cap\alpha_2},\qquad 
X_{\alpha_i} := X_{m_i,\kappa_{i1}\kappa_{i2}}.
\end{equation}
The operator $Z_{\alpha_1\cap\alpha_2}$ has support along the string-like intersection of $m_1$ and $m_2$. By lemma~\ref{lem:outside} its syndrome charges can only reside on cells containing faces of $\phi_1$ or $\phi_2$%
\footnote{
Every element of $E(\phi_1+\phi_2)$ is, up to a stabilizer, the product of an element of $E(\phi_1)$, an element of $E(\phi_2)$ and $z(\phi_1,\phi_2)$, see section~\ref{sec:support}.
},
which only overlap with the support of  $Z_{\alpha_1\cap\alpha_2}$ on its endpoint regions. Therefore $Z_{\alpha_1\cap\alpha_2}$ is a string-like operator. The charge that it transfers can be inferred using the strategy of~\ref{sec:duality}, which requires a third membrane $m_3$ pierced by the string, as in figure~\ref{fig:three_membranes}, and computing, for each color combination,
\begin{equation}\label{eq:intersection}
(Z_{\alpha_1\cap\alpha_2},X_{\alpha_3})=(-1)^{|\alpha_1\cap\alpha_2\cap\alpha_3|},
\end{equation}
where
\begin{equation}
X_{\alpha_3}:=X_{m_3, \kappa_{31}\kappa_{32}}.
\end{equation}

\begin{figure}
\centering
\includegraphics[width=.45\columnwidth]{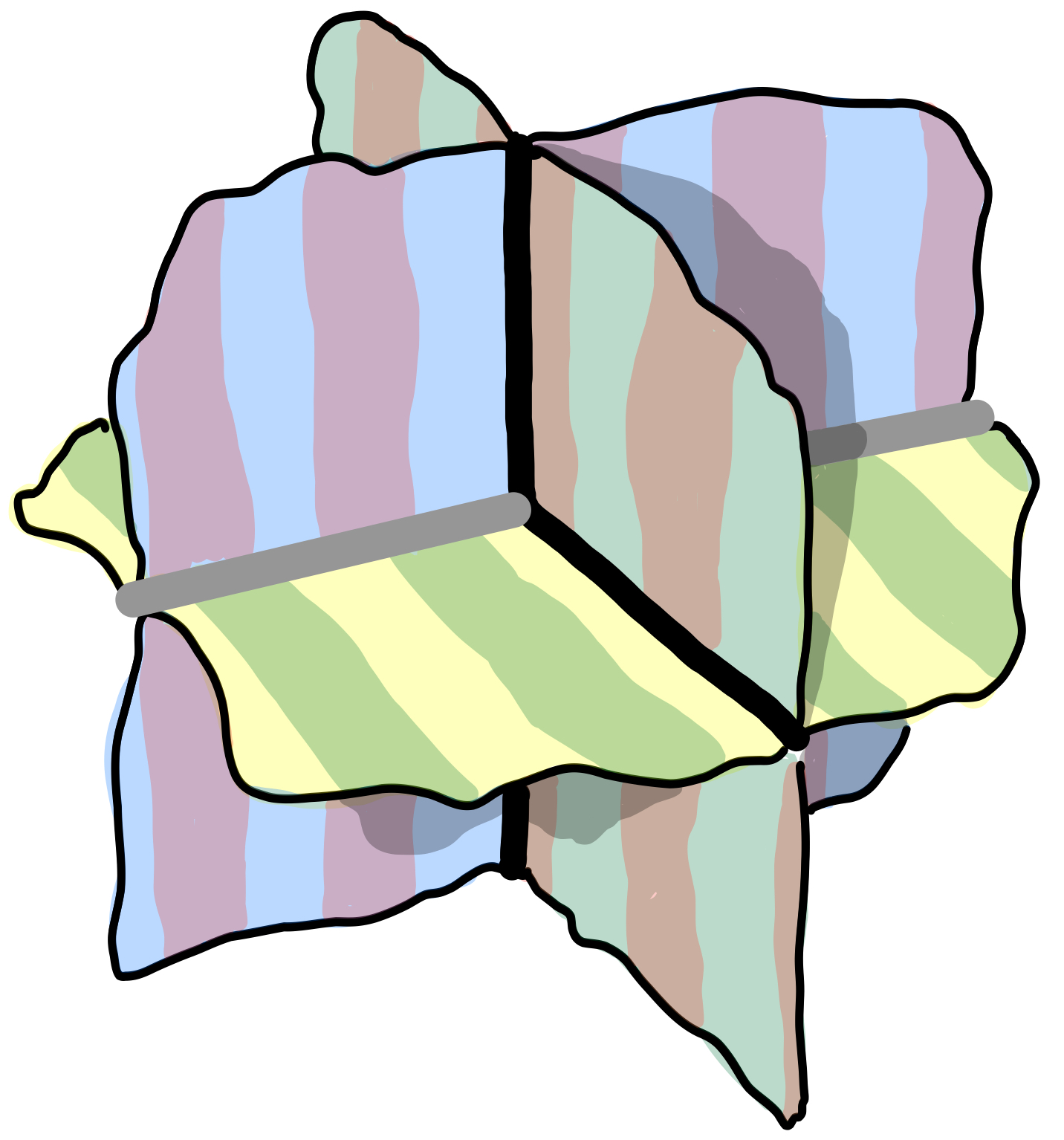}
\caption{The geometry of the three membranes used in the computation of~\eqref{eq:intersection}.}
\label{fig:three_membranes}
\end{figure}

A key aspect of \eqref{eq:intersection} is its topological nature, \emph{i.e.} the transferred charge only depends on the flux labels $\kappa_{11}\kappa_{12}$ and $\kappa_{21}\kappa_{22}$. This is so because (i) we can compute this charge locally anywhere along the intersection of the membranes, (ii) the charge has to be the same along any such string, and (iii) we can do the computation on arbitrary lattices that put together any two given local geometries.

\begin{success}
\center
$z(\phi_1,\phi_2)$ transfers a charge $\lambda(\kappa_{11}\kappa_{12},\kappa_{21}\kappa_{22})$ along the intersection.
\end{success}

\noindent The computation of \eqref{eq:intersection} can thus be perfomed for whatever lattice and membrane geometry is more convenient. The result is that the sign is negative iff the following conditions are satisfied
\begin{equation}
\kappa_{i1}\kappa_{i2}\neq\kappa_{j1}\kappa_{j2},
\qquad i\neq j.
\end{equation}
That is, the charge transferred along the intersection of the membranes is, with $\kappa_i$ all different,
\begin{align}
\lambda(\kappa_1\kappa_2,\kappa_1\kappa_2)&=0,
\\
\lambda(\kappa_1\kappa_2,\kappa_1\kappa_3)&=\kappa_4,
\\
\lambda(\kappa_1\kappa_2,\kappa_3\kappa_4)&=\kappa_1+\kappa_2.
\end{align}

\subsection{Net charge exchange}

Consider arbitrary flux configurations $\phi_i$, $i=1,2$. As long as they are not too cramped together, \emph{i.e.} from a renormalization perspective, they can always be decomposed into some loop-like configurations $\phi_{ij}$
\begin{equation}\label{eq:decompose}
\phi_i = \sum_j \phi_{ij},
\end{equation}
such that $\phi_{ij}$ and $\phi_{i'j'}$ are not connected for $i\neq i'$ and each carry an elementary flux unit, as in the previous section.
Given any $X_{\alpha_{ij}}$ with syndrome $\phi_{ij}$ and $\alpha_i := \sum_j \alpha_{ij}$, the effect of $z(\phi_1,\phi_2)$ can be computed from that of $z(\phi_{1j},\phi_{2k})$ noting that
\begin{equation}
\alpha_1\cap\alpha_2 = \sum_{jk} \alpha_{1j}\cap\alpha_{2k}.
\end{equation}
There is a unique way to extend the above values of $\lambda$ to a morphism
\begin{equation}
\lambda: \text{flux}\times\text{flux}\rightarrow\text{charge}.
\end{equation}
If two loop-like flux configurations carrying charges $h_1$ and $h_2$ have odd linking number they exchange a charge $\lambda(h_1, h_2)$, and if they have even linking number they do not exchange any charge%
\footnote{
The corresponding membranes intersect along a number of string-like regions, each transferring a charge $\lambda(h_1, h_2)$. The parity of the number of such strings connecting the two loops is the same as the parity of the linking number.
}. 

\begin{success}
The charge exchanged by $\phi_1$ and $\phi_2$, given the loop decomposition~\eqref{eq:decompose}, is
\begin{equation}
\sum_{(j,k)\in L} \lambda (h_{1j}, h_{2k})
\end{equation}
where $h_{ij}$ is the flux carried by $\phi_{ij}$, and $L$ is the set of pairs $(j,k)$ such that $\phi_{1j}$ and $\phi_{2j}$ have odd linking number. 
\end{success}

\noindent
By analogy with the linking number, we refer to this exchanged charge as the linking charge of $\phi_1$ and $\phi_2$. It is not a singular phenomenon of 3D color codes: appendix \ref{sec:linking_general} outlines a more general exploration of this topic.

\section {Discussion}

Transversal gates in color codes are a surprisingly rich subject. In addition to the conventional transversal gates that require access to all the qubits, we have introduced here transversal gates that can be performed on lower-dimensional subsets and are of great practical interest~\cite{bombin:2018:colorful}. In studying the propagation of errors, we have found that the transversal T gate defines both a natural set of correctable errors and a `linking charge' for flux excitations.

It would be interesting to characterize the linking charge phenomenon or, more broadly, to classify generalized transversal gates for three-dimensional topological order along the lines discussed in appendix~\ref{sec:linking_general}. It is likely that higher dimensional color codes offer similar insights into the physics of such transformations for higher dimensions.

A question that remains unanswered is the physical origin of the T gate in 3D color codes. If, as posited in appendix~\ref{sec:linking_general}, the T gate is characterized by its linking charge and its trivial action on charge and flux labels, then this must be enough to explain its logical action. The goal is to have a renormalized picture of the T gate based on its physics, in contrast with the microscopic picture, based on a combinatorics and offering no physical insight whatsoever. Eventually, this could be used to define similar gates in other topological codes.

\vspace{\baselineskip}
\noindent
{\bf Acknowledgements.}  I  would  like  to  thank  the  whole  PsiQuantum  fault  tolerance  team, Christopher  Dawson,  Fernando  Pastawski,  Kiran  Mathew,  Naomi Nickerson,  Nicolas  Breuckmann, Andrew Doherty, Jordan Sullivan, and Mihir  Pant for  their  considerable  support and  encouragement. In particular I would like to thank Terry Rudolph,  Nicolas Breuckmann,  Naomi Nickerson,  Fernando Pastawski,  Mercedes Gimeno-Segovia,  Peter  Shadbolt  and  Daniel  Dries  for  many  useful  discussions and/or  very  generous  feedback  at  various  stages  of  this  manuscript.

\appendix

\section{Notation}\label{sec:notation}

\begin{warning}
\begin{itemize}[leftmargin=\baselineskip]
\item
The symmetric difference of sets is denoted $+$.
\item
$\hat U$ is the operator $U\cdot U^\dagger$.
\item
$\partial a$ is the error syndrome of $a$.
\item
$X_q, Z_q$ are the $X$ and $Z$ Pauli operators at qubit $q$. 
\item
$X_\alpha$, $Z_\alpha$, with $\alpha$ (often implicitly) a set of qubits, are\begin{equation}
X_\alpha := \prod_{q\in\alpha} X_q,\qquad
Z_\alpha := \prod_{q\in\alpha} Z_q.
\end{equation}
\item
$P$ is \emph{the} Pauli group, and \emph{a} Pauli group is any of its subgroups.
\item
$P_X$, $P_Z$ are the Pauli groups generated by $X$ and $Z$ operators respectively. 
\item
$a|_x\propto a$ if $a=b\otimes c$ is a Pauli operator on a system $x\otimes y$.
\item
$(a,b)=\pm 1$ is the group commutator of the Pauli operators $a, b$.
\end{itemize}
\vspace{6pt}
\noindent Given sets $A, B$, of Pauli operators:
\begin{itemize}[leftmargin=\baselineskip]
\item
$A|_x$ contains the elements of $A$ restricted to the subsystem $x$, i.e. it is the set $\{a|_x \,|\, a \in A\}$. Notice that $A|_x=\langle i \rangle\, A|_x$.
\item
$A\|_x$ is the subset of elements of $A$ with support in the subsystem $x$, i.e. it is the set $\{a_x \,|\, a_x\otimes \mathbf 1_y \in A\}$.
\item
$\mathcal Z_A(B)$ is the subset of elements of $A$ that commute with the elements of $B$ (with $\mathcal Z_X$ standing for $\mathcal Z_{P_X}$).
\item
$\mathcal L(A)$ is the set of linear combinations of elements of $A$.
\item
$\mathcal D_A$ is the operator
 
\begin{equation}
\mathcal D_{A}:\rho\mapsto \frac 1 {|A|}\sum_{a\in A} a \rho a^\dagger.
\end{equation}

\end{itemize}
\end{warning}

\section{Proofs for section~\ref{sec:T}}\label{sec:Tproof}

The setting for this appendix is the same as in section~\ref{sec:setting}.

\begin{lem}\label{lem:A0}
Given any $X_\alpha\in \mathcal Z(S)$ and a logical $X_\lambda$,
\begin{equation}
H_\alpha = S_Z,\qquad G_\lambda = \mathcal Z_Z(S) = \langle Z_\lambda \rangle S_Z.
\end{equation}
\end{lem}

\begin{proof}[\skproof] For any $X_\beta\in S$, $X_\gamma\in\mathcal Z(S)$
\begin{equation}
(Z_{\alpha\cap\beta}, X_\gamma)= (Z_{\alpha\cap\gamma},X_\beta)=1,
\end{equation}
where the second equality uses $\eqref{axiom_code}$. That is, $Z_{\alpha\cap\beta}\in S$ for any $X_\alpha\in\mathcal Z(S)$ and $X_\beta\in S$, and the above equalities are easy to check.
\end{proof}

\begin{proof}[\skproof\ of lemma~\ref{lem:invariance}]
The first two equations are a trivial consequence of lemma~\ref{lem:A0}. For the third, we need in addition the fact that for any $X_\alpha$, any $X_\gamma\in\mathcal Z(G_\alpha)$ and any $X_\beta\in S$
\begin{equation}
g(\gamma\cap \alpha) \equiv g(\gamma\cap (\alpha+\beta))\mod 4.
\end{equation}
To check this, notice that
\begin{align}
g(\gamma\cap (\alpha+\beta))-g(\gamma\cap \alpha) 
&= g(\gamma\cap\beta)-2g(\gamma\cap\alpha\cap\beta)
\nonumber\\
&\equiv 0
\mod 4,
\end{align}
where the second term vanishes modulo 4 because $(X_\gamma, Z_{\alpha\cap\beta})=1$, and the first by the invariance of the code space under $U$. In particular, since the states $|0\rangle$ and $X_\beta|0\rangle$ pick up the same phase, and the same is true for $X_\gamma|0\rangle$ and $X_{\gamma+\beta}|0\rangle$ because $X_\gamma\in\mathcal Z(S)$, we have
\begin{align}
g(\beta)&\equiv 0\mod 8, \\
g(\beta)-2g(\gamma\cap\beta)&\equiv 0\mod 8. \qedhere
\end{align}
\end{proof}

\begin{lem}\label{lem:A1}
For any $X_\alpha$ and any $X_\beta\in \mathcal Z(S)$,
\begin{align}
Z_{\alpha\cap\beta}\in S&\iff X_\beta\in \mathcal Z(G_\alpha),
\\
Z_{\alpha\cap\beta}\in \mathcal Z(S)&\iff X_\beta\in \mathcal Z(H_\alpha).
\end{align}
\end{lem}

\begin{proof}[\skproof] For the first relation
\begin{align}
Z_{\alpha\cap\beta}\in S
&\iff
\left(
\forall X_\gamma\in\mathcal Z(S) \quad (Z_{\alpha\cap\beta},X_\gamma)=1
\right)
\nonumber\\
&\iff
\left(
\forall X_\gamma\in\mathcal Z(S) \quad (Z_{\alpha\cap\gamma},X_\beta)=1
\right)
\nonumber\\
&\iff 
X_\beta\in \mathcal Z(G_\alpha),
\end{align}
and the second is analogous exchanging $S$ and $\mathcal Z(S)$.
\end{proof}

\begin{lem}\label{lem:A2}
For any $X_\alpha$, the group $H_\alpha$ contains no logical operators.
\end{lem}

\begin{proof} Suppose that $Z_{\alpha\cap\beta}$ is logical for some $X_\beta\in S$. Let $\gamma:=\alpha\cap\beta$. Then $Z_\gamma$ is a logical operator and, by lemma~\ref{lem:A0}, $Z_\gamma\in G_\lambda$ for any logical $X_\lambda$. Since\begin{equation}
Z_{\lambda}=Z_{\lambda\cap\lambda}\in G_\lambda
\end{equation}
we have \begin{equation}
Z_\gamma Z_\lambda \in S_Z
\end{equation}
which gives, by lemma~\ref{lem:invariance} and lemma~\ref{lem:A0}, \begin{equation}
H_\gamma = H_\lambda = S_Z,
\end{equation}
which contradicts\begin{equation}
Z_\gamma = Z_{\gamma\cap\beta}\in H_\gamma.\qedhere
\end{equation}
\end{proof}

\begin{proof}[\skproof\ of lemma~\ref{lem:tolerability}]
{\bf (i)} This follows from lemma~\ref{lem:A1}, noticing for the 'only if' that 
 
\begin{equation}
\mathcal Z(G_\alpha), \mathcal Z(H_\alpha)\subseteq\mathcal Z(S_Z).
\end{equation}{\bf (ii)} The 'if' direction follows from (i), since\begin{equation}
G_\alpha = H_\alpha \langle Z_{\alpha\cap\lambda}\rangle=H_\alpha.
\end{equation}
For the 'only if' direction, choose any logical $X_{\lambda'}$. By (ii)\begin{equation}
H_\alpha=G_\alpha = H_\alpha \langle Z_{\alpha\cap\lambda'}\rangle,
\end{equation}
and thus $Z_{\alpha\cap\lambda'}\in H_\alpha$. Therefore there exists $X_\beta\in S$ such that\begin{equation}
Z_{\alpha\cap\lambda'}Z_{\alpha\cap\beta}\in S_Z,\qquad
\end{equation}
and we can take $\lambda = \lambda'+\beta$. 
{\bf (iii)} This follows from lemma~\ref{lem:A1} and (ii).
{\bf (iv)} By lemma~\ref{lem:invariance} and (i), it is enough to check this for a single logical $X_\lambda$. Notice first that\begin{equation}
G_{\alpha+\lambda} = H_{\alpha+\lambda} \langle Z_{(\alpha+\lambda)\cap\lambda}\rangle
= H_{\alpha} \langle Z_{\alpha\cap\lambda}Z_\lambda\rangle.
\end{equation}
where the second equality is by lemma~\ref{lem:invariance}.
If $X_\alpha$ is tolerable, choosing $\lambda$ as in (ii) gives $Z_\lambda\in G_{\alpha+\lambda}$. By lemma~\ref{lem:A0} $Z_\lambda$ is logical and, by lemma~\ref{lem:A2} and (i), $X_{\alpha+\lambda}$ is not tolerable. 
If $X_\alpha$ is not tolerable, by the definition of tolerability and by lemma~\ref{lem:A2}, we can choose $\lambda$ so that $Z_{\alpha\cap\lambda}$ is logical. Then $Z_{\alpha\cap\lambda}$ is equivalent to $Z_\lambda$, and thus $G_{\alpha+\lambda}=H_\alpha$ and $X_{\alpha+\lambda}$ is tolerable by (i).
{\bf (v)} This is a consequence of (iv), because $X_\beta = X_{\beta+\lambda}X_\lambda.$
\end{proof}

\begin{warning}

For any $X_\alpha$ we define the unitary operator\begin{equation}
A_\alpha := \prod_q T_q^{2 b_q},
\end{equation}
where $T_q$ applies the $T$ gate to the $q$-th qubit. The projector onto the subspace $S_Z$ is\begin{equation}
P_0 := \frac 1 {|S_Z|}\sum_{s\in S_Z} s.
\end{equation} 

\end{warning}

\begin{lem}\label{lem:A3}

For any $X_\alpha$\begin{equation}
A^\dagger_\alpha P_0 \in \mathcal L(E_\alpha).
\end{equation}
\end{lem}

\begin{proof} Given a character (a syndrome) $\nu$ of $\mathcal Z_X(G_\alpha)$ let $P_\nu$ be the projector onto the corresponding syndrome subspace, i.e.\begin{equation}
P_\nu := \frac 1 {|\mathcal Z_X(G_\alpha)|} \sum_{x\in \mathcal Z_X(G_\alpha)} \nu(x)\,x.
\end{equation}
Consider the function\begin{equation}
\sigma(X_\beta):= i^{g(\alpha\cap\beta)},\qquad X_\beta\in \mathcal Z(G_\alpha).
\end{equation}
For any $X_\beta\in \mathcal Z_X(G_\alpha)$\begin{equation}
X_\beta A^\dagger_\alpha X_\beta P_0 = \sigma(X_\beta) Z_{\alpha\cap\beta} A^\dagger_\alpha P_0 =
\sigma(X_\beta) A_\alpha^\dagger P_0,
\end{equation}
where the second equality is by lemma~\ref{lem:A1}. Setting\begin{equation}
a_{\mu,\nu}:=P_\mu A_\alpha^\dagger P_\nu P_0 
\end{equation}
for any $x\in \mathcal Z_X(G_\alpha)$ we get\begin{equation}
a_{\mu,\nu} = \mu(x)\nu(x)P_\mu x A^\dagger_\alpha x P_\nu P_0= 
\sigma(x)\mu(x)\nu(x) a_{\mu,\nu}
\end{equation}
where we have used that $P_0$ and $P_\nu$ commute (because $\mathcal Z_X(G_\alpha)\subseteq \mathcal Z(S_Z)$). Therefore $a_{\mu,\nu}$ is zero unless $\sigma$ is the caracter\begin{equation}
\sigma = \mu\nu.
\end{equation}
Inserting the identity twice we have\begin{equation}
A^\dagger_\alpha P_0 = \sum_{\mu,\nu} a_{\mu,\nu}=\sum_\nu a_{\nu\sigma,\nu}.
\end{equation}
For any $z\in E_\alpha$\begin{equation}
zA^\dagger_\alpha P_0 = \sum_{\nu} z P_{\sigma\nu} A^\dagger_\alpha P_0 P_{\nu} = \sum_{\nu} P_{\nu} z A^\dagger_\alpha P_0 P_{\nu}
\end{equation}
and thus\begin{equation}
zA^\dagger_\alpha P_0 \in \mathcal L(\mathcal Z_Z(\mathcal Z_X(G_\alpha))) = \mathcal L(G_\alpha).
\end{equation}
Finally, since $z^2= \mathbf 1$, using coset notation we have\begin{equation}
A^\dagger_\alpha P_0 \in \mathcal L(zG_\alpha) = \mathcal L(E_\alpha).\qedhere
\end{equation}
\end{proof}

\begin{lem}\label{lem:A4}

Let $\rho$ be an encoded state.
\begin{enumerate}[label=\roman*)]
\item
For any $X_\alpha$ and any $X_\beta\in S\cap \mathcal Z(G_\alpha)$\begin{equation}
\langle A_\alpha X_\beta A_\alpha^\dagger\rangle_\rho
= (-1)^{g(\alpha\cap \beta)/2}.
\end{equation}
\item
For any tolerable $X_\alpha$ and any $X_\beta\in S- \mathcal Z(G_\alpha)$ \begin{equation}
\langle A_\alpha X_\beta A_\alpha^\dagger\rangle_\rho
= 0.
\end{equation}
\end{enumerate}
\end{lem}

\begin{proof}
An easy computation gives for $X_\beta\in S$\begin{equation}
\langle A_\alpha X_\beta A_\alpha^\dagger\rangle_\rho
= i^{g(\alpha\cap \beta)}
\langle X_\beta Z_{\alpha\cap\beta} \rangle_\rho
= i^{g(\alpha\cap \beta)}
\langle Z_{\alpha\cap\beta} \rangle_\rho.
\end{equation}
If $X_\beta\in \mathcal Z(G_{\alpha})$ then $Z_{\alpha\cap\beta}\in S$ by lemma~\ref{lem:A1} and thus\begin{equation}
\langle Z_{\alpha\cap\beta} \rangle_\rho=1.
\end{equation} 
(Notice that $g(\alpha\cap\beta)$ is even because $X_\beta$ is self-adjoint.) If $X_\beta\not\in \mathcal Z(G_{\alpha})$ and $X_\alpha$ is tolerable then by lemma~\ref{lem:A1}\begin{equation}
Z_{\alpha\cap\beta}\not\in \mathcal Z(S)
\end{equation} 
and thus there exists some $s\in S_X$ such that 
\begin{equation}
(s,Z_{\beta\cap\alpha})=-1,
\end{equation}
which gives
\begin{align}
\langle Z_{\alpha\cap\beta} \rangle_{U\rho U^\dagger} = 
\langle Z_{\alpha\cap\beta} s \rangle_{U\rho U^\dagger} 
&= - \langle sZ_{\alpha\cap\beta} \rangle_{U\rho U^\dagger} 
\nonumber\\
&= - \langle Z_{\alpha\cap\beta} \rangle_{U\rho U^\dagger},
\end{align}
i.e. the expectation value vanishes. 
\end{proof}

\begin{lem}\label{lem:A5} 
(Here $S$ is arbitrary.)
Let $S$ be a stabilizer and $G$ a Pauli group such that\begin{equation}
\mathcal Z_G(S) \propto G\cap S.
\end{equation}
If $\rho$ is an encoded state of $S$ such that for certain $K\in \mathcal L(G)$ and any $s\in S$\begin{equation}
\langle K^\dagger s K \rangle_\rho
= \begin{cases} 
1, & \text{if $s\in \mathcal Z_{S}(G)$},\\
0, & \text{otherwise}.
\end{cases}
\end{equation}
then\begin{equation}
\mathcal D_{S}(K\rho K^\dagger)=\mathcal D_G(\rho).
\end{equation}
\end{lem}

\begin{proof}[Sketch of proof] Consider the quotient groups\begin{equation}
S_G := \frac S{\mathcal Z_S(G)},
\qquad
G_S := \frac G {\mathcal Z_G(S)}.
\end{equation}
The irreps (characters) $\mu$ of $S_G$ are morphisms\begin{equation}
\mu: S_G\rightarrow \pm 1.
\end{equation}
Every element $a$ of $G$ is assigned such an irrep via the map\begin{equation}
a\mapsto (a,\cdot)
\end{equation}
It easy to check that this is a biyection from $G_S$ to the group of irreps of $S_G$. That is, $S_G$ and $G_S$ are dual groups.

We choose a set of representatives $k_\mu\in G$ of the quotient group $G_S$, with \begin{equation}
(k_\mu, \cdot)=\mu,
\end{equation}
and proceed in two steps.

\noindent{\bf a)}  First we show that there exist $c_\mu \in \mathbf C$ such that\begin{equation}
\rho':=\mathcal D_S(K\rho K^\dagger) = \sum_\mu c_\mu k_\mu \rho k_\mu^\dagger.
\end{equation}
The state $K\rho K^\dagger$ is a linear combination of terms of the form\begin{equation}
a\rho b^\dagger,\qquad a,b\in G.
\end{equation}
If $ab\not\in \mathcal Z(S)$, there exists $s\in S$ such that $(a,s)=-(b,s)$, so that
\begin{align}
\mathcal D_S(a \rho b^\dagger)
&=\frac 1 2 \mathcal D_S(a \rho b^\dagger+sa\rho b^\dagger s)
\nonumber\\
&=\frac 1 2 \mathcal D_S(a \rho b^\dagger-as\rho s b^\dagger)=0
\end{align}
Otherwise $ab\in \mathcal Z(S)$ and there exists some irrep $\mu$ and $s,s'\in \mathcal Z(S)$ such that\begin{equation}
a=k_\mu s, \qquad b=k_\mu s'.
\end{equation}
Moreover, since $a, b, k_\mu\in G$ we have $s,s'\in \mathcal Z_G(S) \propto G\cap S$ and thus\begin{equation}
\mathcal D_S(a \rho b^\dagger)=
\mathcal D_S(k_\mu s \rho {s'}^\dagger k_\mu^\dagger)\propto
\mathcal D_S(k_\mu \rho k_\mu^\dagger) = k_\mu \rho k_\mu^\dagger.
\end{equation}

\noindent{\bf b)} Armed with the above expresion for $\rho'$ we have for any $s\in S$
\begin{align}
\langle K^\dagger s K \rangle_{\rho} 
= \langle s \rangle_{\rho'} 
&= \sum_\mu c_\mu \langle k_\mu^\dagger s k_\mu\rangle_\rho
\nonumber\\
&= \sum_\mu c_\mu \mu(s) \langle s\rangle_\rho
= \sum_\mu c_\mu \mu(s),
\end{align}
where the sum is over the characters $\mu$ of $S_G$. By a well known property of characters:\begin{equation}
c_\mu = \frac 1{|S_G|}\sum_{s\in S_G} \langle K^\dagger s K\rangle_\rho \,\mu(s)
=\frac 1 {|S_G|}=\frac 1 {|G_S|}.
\end{equation}
Finally, using again that $Z_G(S)\propto S\cap G$,\begin{equation}
\rho' 
= \frac 1 {|G_S|}\sum_\mu k_\mu \rho k_\mu^\dagger
= \frac 1 {|G|}\sum_{g\in G} g \rho g^\dagger.\qedhere
\end{equation}
\end{proof}

\begin{proof}[Proof of theorem~\ref{thm:T}]
We set $X_\alpha = x$ and consider two separate cases.

\noindent {\bf($X_\alpha$ tolerable)} Since $A_\alpha P_0$ cannot be zero, by lemma~\ref{lem:A3} there exist $z\in E_\alpha$ and $K \in\mathcal L(G_\alpha)$ such that\begin{equation}
A_\alpha^\dagger P_0
= z K P_0.
\end{equation}
Combining this with the fact that $U\rho U^\dagger$ is an encoded state of $S_Z$ and using lemma~\ref{lem:A4} we get for any $s\in S$
\begin{align}
\langle K^\dagger s K \rangle_{U\rho U^\dagger}
&= (z, s) \langle A_\alpha s A_\alpha^\dagger \rangle_{U\rho U^\dagger}
\nonumber\\
&=\begin{cases} 
1, & \text{if $s\in Z_S(G_\alpha)\subseteq S_Z \mathcal Z_X(G_{\alpha})$},\\
0, & \text{otherwise}.
\end{cases}
\end{align}
A few manipulations give
\begin{align}
(\mathcal D_{S_X}\circ \hat U\circ \hat X_\alpha) (\rho) 
&= (\mathcal D_S\circ \hat U\circ \hat X_\alpha) (\rho)
\nonumber\\
&= 
(\mathcal D_S \circ \hat X_\alpha \circ \hat U \circ \hat A^\dagger_\alpha) (\rho)
\nonumber\\
&= (\hat X_\alpha \circ \hat z \circ\mathcal D_S \circ \hat K\circ \hat U) (\rho).
\end{align}
Since $X_\alpha$ is tolerable\begin{equation}
\mathcal Z_{G_\alpha}(S) = G_\alpha\cap S
\end{equation}
and thus we can apply lemma~\ref{lem:A5} to the encoded state $U\rho U^\dagger$, recovering the result.

\noindent {\bf($X_\alpha$ not tolerable)} By lemma~\ref{lem:tolerability} we can choose a logical operator $X_\lambda\in\mathcal Z(H(\phi))$. By lemma~\ref{lem:tolerability} (iv), $X_{\alpha+\lambda}$ is tolerable. For $w$ we can choose any logical operator that does the job, and in particular we set\begin{equation}
w = U^\dagger X_\lambda U X_\lambda.
\end{equation}
The result follows from the first case: since $X_\lambda \rho X_\lambda$ is an encoded state
\begin{align}
(\mathcal D_{S_X}\circ \hat U\circ \hat X_\alpha) (\rho) 
&= (\mathcal D_{S_X}\circ \hat U\circ \hat X_{\alpha+\lambda}\circ\hat X_\lambda) (\rho)
\nonumber\\
&= (\hat X_{\alpha+\lambda} \circ \mathcal D_{E_{\alpha+\lambda}}\circ \hat U \circ \hat X_\lambda)(\rho) 
\nonumber\\
&= (\hat X_{\alpha} \circ \mathcal D_{E_{\alpha+\lambda}}\circ \hat U \circ \hat w)(\rho),
\end{align}
where, in the last equality, $\hat X_\lambda$ commutes with $\mathcal D_{E_{\alpha+\lambda}}$ because $G_{\alpha+\lambda}=H_{\alpha}$.
\end{proof}

\begin{proof}[\skproof\ of lemma~\ref{lem:factorization}]
{\bf (i)} Trivially
\begin{equation}
G_\alpha\subseteq \prod_i G_{\alpha_i},
\end{equation}
and, via lemma~\ref{lem:tolerability}, 
\begin{equation}
G_{\alpha_i}=H_{\alpha_i}\subseteq H_\alpha\subseteq G_\alpha.
\end{equation}
The result follows, again via lemma~\ref{lem:tolerability}, because\begin{equation}
G_\alpha = \prod_i G_{\alpha_i} = \prod_i H_{\alpha_i}=H_\alpha.
\end{equation}
{\bf (ii)} This follows combining 
\begin{equation}
\mathcal Z(G_\alpha)= \bigcap_i \mathcal Z(G_{\alpha_i}),
\end{equation}
with
\begin{equation}
\frac {g(\alpha\cap \gamma)} 2
= \sum_i \frac {g(\alpha_i\cap \gamma)} 2 - \sum_{i\neq j} g(\alpha_i\cap\alpha_j\cap \gamma)
\end{equation}
and, for any $\beta$,
\begin{equation}
g(\beta) \equiv |\beta|
\mod 2.\qedhere
\end{equation}
\end{proof}

\begin{proof}[Proof of lemma~\ref{lem:separation}]
Let $i=1,\dots,n$ and $j=1,\dots,m$. The condition on the generators $X_{\beta_j}$ amounts to the existence of a function
\begin{equation}
I:\{1,\dots,m\}\rightarrow \{1,\dots,n\}
\end{equation}
such that
\begin{equation}
i\neq I(j) \implies Z_{\alpha_i\cap \beta_j}\in S.
\end{equation}
Denote by $I^{-1}[i]$ the preimage of $i$. When $j\in I^{-1}[i]$ there exists some $s\in S_Z$ such that
\begin{equation}
Z_{\alpha_i\cap \beta_j}
= s \prod_{k=1}^n Z_{\alpha_k\cap \beta}
= s Z_{\alpha\cap\beta_j}.
\end{equation}
Then:
\begin{align}
H_{\alpha_i} 
&= S_Z\cdot\langle Z_{\alpha_i\cap\beta_j} \rangle_{j=1}^m 
= S_Z\cdot\langle Z_{\alpha_i\cap\beta_j} \rangle_{j\in I^{-1}[i]} 
\nonumber\\
&= S_Z\cdot\langle Z_{\alpha\cap\beta_j} \rangle_{j\in I^{-1}[i]} 
\subseteq S_Z\cdot\langle Z_{\alpha\cap\beta_j} \rangle_{j=1}^k = H_\alpha.
\end{align}
The converse, i.e.
\begin{equation}
H_\alpha \subseteq \prod_i H_{\alpha_i},
\end{equation}
is trivial. 
\end{proof}

\begin{proof}[\skproof\ of lemma~\ref{lem:A6}]
 Trivially 
 \begin{equation}
G_\alpha G_{\alpha+\omega} = G_\alpha G_{\omega},
\end{equation}
and given
\begin{equation}
X_\gamma\in \mathcal Z(G_\alpha G_{\alpha+\omega})=\mathcal Z(G_\alpha)\cap \mathcal Z(G_\omega)\cap\mathcal Z(G_{\alpha+\omega}),
\end{equation}
we have
\begin{align}
(X_\gamma, e_{\alpha}e_{\alpha+\omega}) 
&=
(-1)^{g(\gamma\cap\alpha)/2+g(\gamma\cap(\alpha+\omega))/2}
\nonumber\\&=
(-1)^{g(\gamma\cap\omega)/2+|\gamma\cap\alpha\cap\omega|}
\\&=
(X_\gamma, e_\omega Z_{\alpha\cap\omega}).\qedhere
\end{align}
\end{proof}

\section{Gates in tetrahedral codes}\label{sec:tetrahedral}

This appendix discusses the implementation of transversal gates in tetrahedal codes. For the CNot gate see~\cite{bombin:2007:3dcc}.

\subsection{T gate}

The vertices of a tetrahedral colex are bicolorable, i.e. we can partition them in two sets so that vertices connected by an edge belong to different sets. Given such a bipartion, we assign accordingly to each qubit $q$ a sign
\begin{equation}
b_q = \pm 1.
\end{equation}
The logical $T$ gate has the transversal implementation~\cite{bombin:2015:gauge}:
\begin{equation}\label{eq:logical_T}
U = \bigotimes_q T^{k_0b_q}, \qquad k_o:\equiv \sum_q b_q \mod 8.
\end{equation} 

\subsection{P gate}\label{sec:P}

The logical $P$ gate can be implemented as $U^2$, with $U$ as above, but this is not the only way. Clearly, given any stabilizer $X_\alpha$, with $\alpha$ some set of qubits, the gate $X_\alpha U X_\alpha$ implements the logical $T$ gate. Therefore the $P$ gate has the transversal implementation
\begin{equation}
U X_\alpha U X_\alpha \propto \prod_{q\in\bar\alpha} T_q^{2k_0b_q},
\end{equation}
where $T_q$ is the $T$ gate on the $q$-th qubit and $\bar\alpha$ is the complement of $\alpha$. $X_\alpha$ is a stabilizer iff $X_{\bar\alpha}$ implements the logical $X$ operator. Thus, we can choose $P$ to have the same support as any logical $X$ operator, and these operators are membrane like.

\subsection{Gates on facets}\label{sec:gates_facets}

CP gates can be implemented facet-to-facet. This is due to the matryoshka-like relationship between simplicial color codes~\cite{bombin:2013:self} of different dimensions, of which triangular codes~\cite{bombin:2006:2dcc} and tetrahedral codes are the low dimension cases. 

Let $S=S_XS_Z$ be the stabilizer of the (3D) tetrahedral code and $S^r=S^r_XS^r_Z$ the stabilizer of the (2D) triangular code on a facet $r$. Then
\begin{equation}\label{eq:SXr}
S_X|_r\propto S_{X}^r,\qquad
\mathcal Z_X(S)|_r\propto \mathcal Z_{P_X|_r}(S^r),
\end{equation}
or dually%
\footnote{Using the relation $\mathcal Z(A)|_r = \mathcal Z_{r}(A||_r)$, see~\cite{bombin:2018:colorful}.}
\begin{equation}\label{eq:SZr}
S_Z\|_r = S_{Z}^r,\qquad
\mathcal Z_Z(S)\|_r = \mathcal Z_{P_Z|_r}(S^r).
\end{equation}
In fact, it is not only the $CP$ gate that can be performed on the boundary. On a tetrahedral colex of dimension $D$ we can define a color code such that on a boundary element of dimension $d$ a gate of the $d$-th level of the Pauli hierarchy is transversal (the first level being the Pauli group, the second the Clifford group, and so on). The general observation is as follows.

\begin{lem}\label{lem:subsystem}
Let $S=S_XS_Z$ be a CSS code defined on a set of qubits, and $S^r=S^r_XS^r_Z$ another CSS code defined on a subset $r$ of those qubits such that the above relations are satisfied. Let $U$ be a unitary on the subsystem $r$ that commutes with $Z$ operators. If $U$ implements a logical gate $\bar U$ on the second code, then it also implements $\bar U$ on a logical subsystem of the first one.
\end{lem}

\begin{proof}[\skproof] The encoded computational states of the second code take the form, for some vector space $V$ and linear map $f$
\begin{equation}
|\bar x\rangle = \sum_{y\in V} |y+f(\bar x)\rangle.
\end{equation}
Each $f(\bar x)$ belongs to a different representative of the quotient space over $V$. Thus, up to a phase, $U$ takes the form
\begin{equation}
U|y+f(\bar x)\rangle = \phi(\bar x)|y+f(\bar x)\rangle, \qquad y\in V.
\end{equation}
That is, $\bar U$ is also diagonal on the logical computational base $\bar x$ and has entries $\phi(\bar x)$.

For the first code the physical qubits are divided in the subset $r$ and its complement $\bar r$. It is easy to derive from the assumptions that the logical qubits are divided in two subsystems, and there is some subspace $W$ and a linear map $f'$ such that
\begin{equation}
|\bar x\oplus \bar x'\rangle = \sum_{y\oplus y'\in W} |y+f(\bar x)\rangle|y'+f'(\bar x\oplus \bar x')\rangle,
\end{equation}
where $y\in V$ for every $y\oplus y'\in W$. Then, as stated,
\begin{equation}
U\otimes \mathbf 1|\bar x\oplus \bar x'\rangle = \phi(\bar x)|\bar x\oplus \bar x'\rangle.\qedhere
\end{equation}
\end{proof}

\section{Clifford gates and errors}\label{sec:noise_Clifford}

This appendix discusses error propagation for logical Clifford gates on tetrahedal codes.

\subsection{General approach}

Clifford gates preserve the Pauli group, and this makes much easier to study the propagation of Pauli errors. Offen the computation is straightfoward, but sometimes it is faster to take a different route. A transversal Clifford gate induces a mapping on syndromes
\begin{equation}
(\xi,\phi)\mapsto (\xi',\phi').
\end{equation}
If any given Pauli operator $p$ has support on a set of qubits not containing the support of any logical operator, one can:
\begin{itemize}
\item
compute the syndrome $(\xi,\phi)$ of $p$,
\item
obtain from it the syndrome $(\xi',\phi')$ of $UpU^\dagger$, and
\item
recover $UpU^\dagger$ as the unique error, up to stabilizers, with syndrome $(\xi',\phi')$ and support contained in the support of $p$.
\end{itemize}
The procedure can be applied separately to arbitrary factors of a Pauli error.

\subsection{Standard P gate}

The `standard' P gate is the square of the logical T gate defined in~\ref{eq:logical_T}. It leaves $Z$ operators invariant, and transforms an $X$ \emph{face} operator as follows
\begin{equation}
U X_f U^\dagger = X_f Z_f.
\end{equation}
For any cell $c$, of some color $\kappa$, choose a color $\kappa'\neq \kappa$ and let $F$ the set of faces that $c$ shares with $\kappa'$ colored cells. The cell operator takes the form
\begin{equation}
X_c = \prod_{f\in F} X_f.
\end{equation}
Thus we have
\begin{equation}
U\left( X_c \prod_{f\in F} Z_f \right)U^\dagger = X_f,
\end{equation}
which leads to
\begin{equation}
(\xi,\phi) \mapsto (\xi + \text{br}(\phi), \phi),
\end{equation}
with $\text{br}$ as in the main text.

\subsection{Facet P gate}

This form of the $P$ gate is obtained restricting the unitary $U$ of the previous section to a given facet $r$, see appendix~\ref{sec:tetrahedral}. For cells not in contact with the facet $r$ the cell operator $X_c$ does not change. If a cell $c$ has qubits on $r$ then the intersection is a face $f$ and thus
\begin{equation}
U X_c Z_f U^\dagger = X_f,
\end{equation}
which leads to
\begin{equation}
(\xi,\phi) \mapsto (\xi + \text{end}_r(\phi), \phi),
\end{equation}
with $\text{end}_r$ as in the main text.

\section{Proofs for section~\ref{sec:propagation}}\label{sec:noise_T}

Throughout this appendix $S=S_XS_Z$ is the stabilizer of a tetrahedral code.

\begin{warning}
Given an $X$-error syndrome $\phi$, $W_\phi$ is the set of cells and facets with at least one of their faces in $\phi$,
\end{warning}

\begin{lem}\label{lem:syndrome}
Given $X_\alpha$ with syndrome $\phi$ (i) for any cell $c$
\begin{equation}
c\not\in W_\phi
\iff
Z_{\alpha\cap c}\in S,
\end{equation}
and (ii) for any facet $r$\begin{equation}
r\not\in W_\phi
\iff
Z_{\alpha\cap r}\in \mathcal Z(S),
\end{equation}
\end{lem}

\begin{proof}[\skproof] Any cell or facet $w$ can be regarded as a 2D color code with stabilizer $S^w$. For facets, we already encountered it in appendix~\ref{sec:gates_facets}.
In the case of cells there are no logical qubits, but the same relations \eqref{eq:SXr} and \eqref{eq:SZr} hold. Since $S^w_Z$ is generated by the face operators of $w$, the condition $w\not\in W_\phi$ is equivalent to
\begin{equation}
X_w\in \mathcal Z(S^w),
\end{equation}
which in turn, since $S^w$ is self-dual, is equivalent to
\begin{equation}
Z_{w\cap\alpha}\in \mathcal Z_K(S^w), \qquad K := P_Z|_w.
\end{equation}
The result follows using~\eqref{eq:SZr}.
\end{proof}

\begin{proof} [\skproof\ of lemma~\ref{lem:globality}]
Take any $X_\alpha$ with syndrome $\phi$ and that commutes with $Z_r$%
\footnote{If $x$ anticommutes with $Z_r$, $xX_r$ commutes with $Z_r$ and $\partial x = \partial (xX_r)$.
}. 
Notice that
\begin{equation}
(Z_{\alpha\cap r}, X_r)=(X_\alpha, Z_r)=1.
\end{equation}
Since $X_r$ is a non-trivial logical operator, by lemma~\ref{lem:syndrome}
\begin{equation}
Z_{\alpha\cap r} \in S.
\end{equation}
The facet can be regarded as a 2D color code with stabilizer $S^r$ as in section~\ref{sec:gates_facets}. By the first equation of~\eqref{eq:SZr} there exist faces $f_i$ in $r$ such that
\begin{equation}
Z_{\alpha\cap r} = \prod_i Z_{f_i}.
\end{equation}
By the first equation of~\eqref{eq:SXr} for each face $f_i$ there exists $x_i\in S_X$ such that%
\footnote{
In fact, $x_i$ can be the cell operator of the only cell with $f_i$ as a face.
} 
taking
\begin{equation}
x =  X_\alpha \prod_i x_i
\end{equation}
gives as desired
\begin{equation}
x|_r \propto X_{\alpha\cap r}\prod_i X_{f_i} = 1.\qedhere
\end{equation}
\end{proof}

\begin{proof}[Proof of lemma~\ref{lem:separability_cc}] Choose any $X_{\alpha_i}$ with syndrome $\phi_i$. It suffices to show that the condition of lemma~\ref{lem:separation} holds. Indeed, cell operators generate $S_X$ and, by lemma~\ref{lem:syndrome}, for any cell $c$ there exists a single $\alpha_i$ with
\begin{equation}
Z_{\alpha_i\cap c}\not\in S.\qedhere
\end{equation} 
\end{proof}

\begin{lem}\label{lem:tree}
Given a color code without boundaries, a $Z$-error syndrome $\xi$ and a tree $t$ on the colex that contains at least a qubit for each of the cells in $\xi$, there exists $z\in P_Z$ with syndrome $\xi$ and support on $t$.
\end{lem}

\begin{proof}[\skproof] We proceed by induction on the number of qubits (vertices) of the tree $t$. Base case: If the tree contains a single qubit $q$ then, by charge conservation, either $\xi$ is empty and $z=1$, or $\xi$ contains the four neighbors of $q$ and $z=Z_q$. Inductive step: If the tree has $n>1$ qubits, choose a leaf $q$ and let $t'$ be the tree obtained from $t$ by removing $q$. There is at most a single cell $c$ that contains $q$ but no qubits of $t'$. If there is no such $c$ or $c\not\in \xi$ the result follows by the inductive assumption. Otherwise the syndrome
\begin{equation}
\xi' = \xi + \partial Z_q
\end{equation}
is such that, by inductive assumption, there exists $z'\in P_Z$ with syndrome $\xi'$ and support on $t'$, and we take $z=z'Z_q$.
\end{proof}

\begin{lem}\label{lem:tree2}
Given a tetrahedral color code, an error $z'\in P_Z$ and a tree $t$ that contains at least
\begin{itemize}
\item
a qubit for each of the cells in $\partial z'$, and
\item
a qubit for each of the facets $r$ such that
\begin{equation}
(X_r, z') = -1,
\end{equation}
\end{itemize}
there exists $z\in P_Z$ with support in $t$ and such that
\begin{equation}
zz' \in S_Z.
\end{equation}
\end{lem}

\begin{proof}[\skproof] 
By adding a single qubit to a tetrahedral code it is possible to make it into a trivial (spherical) code with no boundaries~\cite{bombin:2007:3dcc}. The new code only contains four new cells, one per facet of the tetrahedron. This result is then just an application of lemma~\ref{lem:tree}.
\end{proof}

\begin{lem}[Restatement of lemma~\ref{lem:outside}]
For any $X$-error syndrome $\phi$,
\begin{equation}
\langle X_w \,|\, w\not\in W_\phi \rangle\subseteq \mathcal Z(E(\phi)).
\end{equation}
\end{lem}

\begin{proof}[\skproof] Let $X_\alpha$ be tolerable and have syndrome $\phi$. From lemma~\ref{lem:syndrome}, for $w\not\in W_\phi$, since $X_\alpha$ is tolerable
\begin{equation}
Z_{w\cap\alpha}\in S.
\end{equation} 
By lemma~\ref{lem:A1}, 
\begin{equation}
X_w\in \mathcal Z(G_\alpha),
\end{equation} 
so that for any $z\in E_\alpha$\begin{equation}
(z,X_w)=(-1)^{g(\alpha\cap\omega)/2}.
\end{equation}
Finally, $Z_{w\cap\alpha}$ is an element of $S^w$, the stabilizer of the 2D color code of $w$ (see the proof of lemma~\ref{lem:syndrome}), and thus the above phase is trivial~\cite{bombin:2015:gauge}.
\end{proof}

{\bf Lemma~\ref{lem:Zerror}} follows combining lemmas~\ref{lem:outside} and~\ref{lem:tree2}.

\section{Local gates in 3D topological codes}\label{sec:linking_general}

Local gates are%
\footnote{
A much more proper but lengthy name is locality-preserving unitaries~\cite{beverland:2016:protected}.
},
for the purpose of this appendix, unitaries that generalize transversal gates via the following key property: they do not spread errors beyond a certain fixed distance. Here I take an \emph{intuitive} look at the physics of local gates in 3D topological models for which excitations can be understood in terms of topologically interacting, localized, charges and fluxes.

\subsection{Invariant operator algebras}

Given a region $R$, let $\mathcal B_R$ be the algebra of operators with support on $R$, and $P_R$ a projector onto states with no excitations in a neighborhood of $R$. The neighborhood should be large enough so that the ground state subspace is invariant under the algebra
\begin{equation}
\mathcal A_R:=\{P_RbP_R\,|\,b\in \mathcal B_R\}.
\end{equation}

We are interested in regions $R$ with the topology of a 'thick' loop or a `thick' closed 2-manifold. The geometry of these regions should be defined on a scale much larger than (i) that relevant to excitations, and (ii) that relevant to the local unitaries of interest. The expectation is that if $R'$ is a yet thicker version of $R$ then
\begin{equation}
\mathcal A_{R'}=\{P_{R'}aP_{R'}\,|\,a\in \mathcal A_R\},
\end{equation}
and this correspondence between $\mathcal A_{R'}$ and $\mathcal A_R$ is one-to-one. Then for any local unitary $U$ we can choose $R'$ thick enough so that
\begin{equation}
\mathcal A_{R'} = \{P_{R'}UaU^\dagger P_{R'}\,|\,a\in \mathcal A_R\}.
\end{equation}
That is, $U$ induces an automorphism of $\mathcal A_R$. Such automorphisms are a great tool for investigating local gates~\cite{beverland:2016:protected}.

\subsection{Symmetries}

Let us focus on regions $R$ contained in the bulk of the system. We expect that 
\begin{itemize}
\item
if $R$ is loop-like, in particular the boundary of a disc-like region, then $\mathcal A_R$ is generated by the projectors $P_f$ onto states with a total flux $f$ flowing through the disc, and
\item
if $R$ is sphere-like, in particular the boundary of a ball-like region, then $\mathcal A_R$ is generated by the projectors $P_c$ onto states with a total charge $c$ inside the ball.
\end{itemize}
For such algebras~\cite{beverland:2016:protected} the automorphism induced by a local gate $U$ takes the form of a \emph{permutation} of the flux / charge labels. The permutation is independent of the loop/sphere under consideration, as can be easily derived by taking into account that the vacuum sector is always mapped to itself for any given loop/sphere%
\footnote
{\emph{E.g.} given two balls, one containing the other, consider states with excitations only inside the inner ball or outside the outer one.
}. Moreover, the permutation of charges and fluxes has to be a \emph{symmetry} of the abstract charge and flux model: topological interactions and the fusion of charges or fluxes remain invariant under the permutation.

Every local unitary $U$ induces a symmetry of the topological charge and flux model. Is this enough to characterize local gates, at least inasmuch as bulk properties are concerned? Definitely no%
\footnote{
That is, unless a symmetry is meant to encompass the linking charge effects discussed here. 
}:
the transversal T gate of 3D color codes induces a trivial symmetry and, yet, has nontrivial effects on bulk excitations.

\subsection{Linking charge}

In the previous section we considered loop-like and sphere-like regions $R$, but concluded that the information that can be gathered from them is limited. A natural step to overcome this difficulty is to consider 2-manifold-like regions $R$ with higher genus. In particular, we consider regions $R$ that are the boundary of some 'solid' region, such as a solid torus.

What is $\mathcal A_R$ for such a region $R$? Unlike in a sphere, we can now find loops within $R$ capable of enclosing non-trivial flux configurations (that avoid entering inside $R$).
Let us assume that states with no excitations in $R$ are divided into sectors $f$ where for every such loop the flux is well defined.
Let $P_{f}\in \mathcal A_R$ be the projectors onto those sectors. 
We assume that the algebras
\begin{equation}
\mathcal A_{R,f} := \{P_f a P_f \,|\,a\in \mathcal A_R\},
\end{equation}
are generated by the projectors $P_c$ onto sectors with a total charge $c$ in the region enclosed by $R$.

Consider a local gate $U$ inducing a trivial permutation of both charges and fluxes, in the sense of the previous section. Such a local gate $U$ induces an automorphism of $\mathcal A_{R,f}$. In particular $U$ induces, for each flux sector $f$, a permutation $\pi_f$ of the charge enclosed by $R$. As in the previous section, this permutation is the same for any two regions $R$ and flux configurations $f$ that can be smoothly deformed into each other. The all important difference with respect to the permutations in the previous section is that, for general values of $f$, there is no reason for the trivial sector to remain invariant. This motivates defining the \emph{linking charge}
\begin{equation}
\lambda_f := \pi_f (1)
\end{equation}
where $1$ represent the trivial charge. As we argue next, the values $\lambda_f$ on a torus completely fix $\pi_f$, not only for the torus but also for higher genus surfaces.

\begin{figure}[h]
\includegraphics[width=\columnwidth]{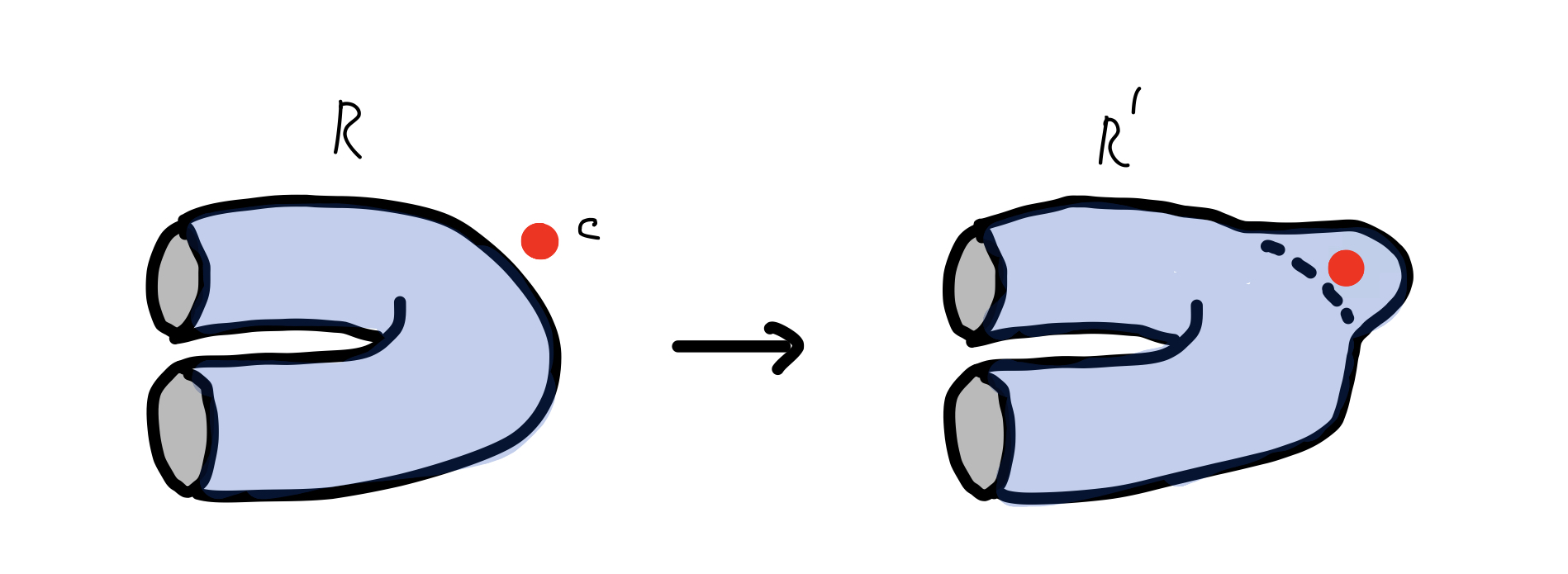}
\caption{Modifying a surface-like region $R$ by adding a `bump' to the region enclosed by $R$.}
\label{fig:bump}
\end{figure}

Consider a slight deformation $R'$ of our region $R$, obtained by adding a 'bump' to the region enclosed by $R$, see figure~\ref{fig:bump}. The flux configuration sectors $f$ of $R'$ and $R$ are trivially identified: these sectors can be defined using loops that avoid the subregion where $R'$ and $R$ differ. For a given sector $f$, consider a state such that $R$ encloses a trivial charge but $R'$ encloses a charge $c$. Since the charge $c$ is contained within the bump, which is sphere-like, it does not change under $U$. Then we can compute in two different ways the charge inside $R'$ after $U$ is applied to such a state, which yields:
\begin{equation}
\pi_f(c)= \lambda_f\times c,
\end{equation}
where $\times$ denotes charge fusion. We conclude that $\lambda_f$ is abelian (it has well defined fusion with any other charge) and that the permutation adds a charge $\lambda_f$.

To see that the linking charge of a torus fixes the linking charge for higher genus surfaces, consider the  two geometries of figure~\ref{fig:tori}. Any flux configuration that avoids the surface on the first geometry (left) can be deformed into a flux configuration that avoids the two tori of the second geometry. In going from the first geometry to the second, we potentially lose information about the flux sectors $f$. Since the whole deformation of the flux can be done in the region bounded by the first surface, it has to be the case that this information about flux sectors is irrelevant to compute the linking charge: by adding the linking charges computed for each torus separately, we can recover the linking charge for the first surface.

\begin{figure}
\includegraphics[width=\columnwidth]{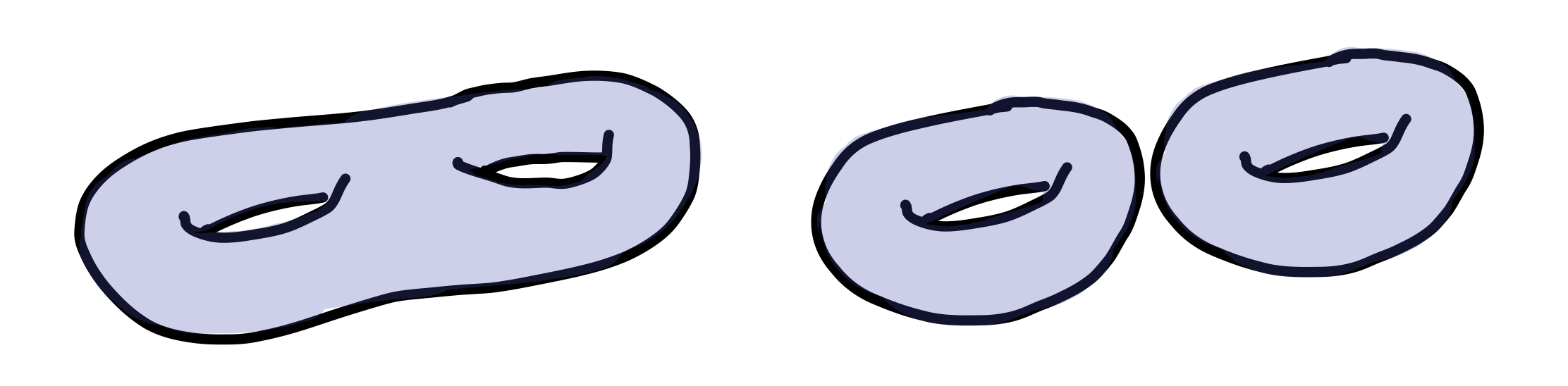}
\caption{A genus-2 region and two tori.}
\label{fig:tori}
\end{figure}

\subsection{Towards a classification of local gates}

It is unclear if symmetries and linking charges are enough to characterize local gates, at least regarding their action in the bulk. We can, however, consider as an example 3D color codes. There is a gate that induces both a trivial symmetry and trivial linking charges: the transversal phase gate P. However, it turns out that P can be performed in a 2D sub-region of the system. Thus a potential scenario consists of a hierarchy of gates, where
\begin{itemize}
\item
truly 3D gates can be characterized by symmetry and linking charges,
\item
some gates can be performed in 2D (and maybe characterized with the same tools as local gates on 2D TQFTs), and
\item
1D gates are string operators carrying abelian charge.
\end{itemize}
This is by analogy with the case of 2D TQFTs, where local gates are fully characterized, up to abelian string operators, by the symmetry transformation they induce%
\footnote{
This can be shown with arguments analogous to the ones presented in this appendix.
}.

\bibliography{refs}

\begin{thebibliography}{13}%
\makeatletter
\providecommand \@ifxundefined [1]{%
 \@ifx{#1\undefined}
}%
\providecommand \@ifnum [1]{%
 \ifnum #1\expandafter \@firstoftwo
 \else \expandafter \@secondoftwo
 \fi
}%
\providecommand \@ifx [1]{%
 \ifx #1\expandafter \@firstoftwo
 \else \expandafter \@secondoftwo
 \fi
}%
\providecommand \natexlab [1]{#1}%
\providecommand \enquote  [1]{``#1''}%
\providecommand \bibnamefont  [1]{#1}%
\providecommand \bibfnamefont [1]{#1}%
\providecommand \citenamefont [1]{#1}%
\providecommand \href@noop [0]{\@secondoftwo}%
\providecommand \href [0]{\begingroup \@sanitize@url \@href}%
\providecommand \@href[1]{\@@startlink{#1}\@@href}%
\providecommand \@@href[1]{\endgroup#1\@@endlink}%
\providecommand \@sanitize@url [0]{\catcode `\\12\catcode `\$12\catcode
  `\&12\catcode `\#12\catcode `\^12\catcode `\_12\catcode `\%12\relax}%
\providecommand \@@startlink[1]{}%
\providecommand \@@endlink[0]{}%
\providecommand \url  [0]{\begingroup\@sanitize@url \@url }%
\providecommand \@url [1]{\endgroup\@href {#1}{\urlprefix }}%
\providecommand \urlprefix  [0]{URL }%
\providecommand \Eprint [0]{\href }%
\providecommand \doibase [0]{http://dx.doi.org/}%
\providecommand \selectlanguage [0]{\@gobble}%
\providecommand \bibinfo  [0]{\@secondoftwo}%
\providecommand \bibfield  [0]{\@secondoftwo}%
\providecommand \translation [1]{[#1]}%
\providecommand \BibitemOpen [0]{}%
\providecommand \bibitemStop [0]{}%
\providecommand \bibitemNoStop [0]{.\EOS\space}%
\providecommand \EOS [0]{\spacefactor3000\relax}%
\providecommand \BibitemShut  [1]{\csname bibitem#1\endcsname}%
\let\auto@bib@innerbib\@empty
\bibitem [{\citenamefont {Lidar}\ and\ \citenamefont
  {Brun~(editors)}(2013)}]{lidar:2013:quantum}%
  \BibitemOpen
  \bibfield  {author} {\bibinfo {author} {\bibfnamefont {D.}~\bibnamefont
  {Lidar}}\ and\ \bibinfo {author} {\bibfnamefont {T.}~\bibnamefont
  {Brun~(editors)}},\ }\href@noop {} {\emph {\bibinfo {title} {Quantum Error
  Correction}}}\ (\bibinfo  {publisher} {Cambridge University Press},\ \bibinfo
  {address} {New York},\ \bibinfo {year} {2013})\BibitemShut {NoStop}%
\bibitem [{\citenamefont {Dennis}\ \emph {et~al.}(2002)\citenamefont {Dennis},
  \citenamefont {Kitaev}, \citenamefont {Landahl},\ and\ \citenamefont
  {Preskill}}]{dennis:2002:tqm}%
  \BibitemOpen
  \bibfield  {author} {\bibinfo {author} {\bibfnamefont {E.}~\bibnamefont
  {Dennis}}, \bibinfo {author} {\bibfnamefont {A.}~\bibnamefont {Kitaev}},
  \bibinfo {author} {\bibfnamefont {A.}~\bibnamefont {Landahl}}, \ and\
  \bibinfo {author} {\bibfnamefont {J.}~\bibnamefont {Preskill}},\ }\href@noop
  {} {\bibfield  {journal} {\bibinfo  {journal} {J. Math. Phys.}\ }\textbf
  {\bibinfo {volume} {43}},\ \bibinfo {pages} {4452} (\bibinfo {year}
  {2002})}\BibitemShut {NoStop}%
\bibitem [{\citenamefont {Bombin}\ \emph {et~al.}(2013)\citenamefont {Bombin},
  \citenamefont {Chhajlany}, \citenamefont {Horodecki},\ and\ \citenamefont
  {Martin-Delgado}}]{bombin:2013:self}%
  \BibitemOpen
  \bibfield  {author} {\bibinfo {author} {\bibfnamefont {H.}~\bibnamefont
  {Bombin}}, \bibinfo {author} {\bibfnamefont {R.~W.}\ \bibnamefont
  {Chhajlany}}, \bibinfo {author} {\bibfnamefont {M.}~\bibnamefont
  {Horodecki}}, \ and\ \bibinfo {author} {\bibfnamefont {M.}~\bibnamefont
  {Martin-Delgado}},\ }\href@noop {} {\bibfield  {journal} {\bibinfo  {journal}
  {New J. Phys.}\ }\textbf {\bibinfo {volume} {15}},\ \bibinfo {pages} {55023}
  (\bibinfo {year} {2013})}\BibitemShut {NoStop}%
\bibitem [{\citenamefont
  {Bombin}(2015{\natexlab{a}})}]{bombin:2015:single-shot}%
  \BibitemOpen
  \bibfield  {author} {\bibinfo {author} {\bibfnamefont {H.}~\bibnamefont
  {Bombin}},\ }\href@noop {} {\bibfield  {journal} {\bibinfo  {journal} {Phys.
  Rev. X}\ }\textbf {\bibinfo {volume} {5}},\ \bibinfo {pages} {031043}
  (\bibinfo {year} {2015}{\natexlab{a}})}\BibitemShut {NoStop}%
\bibitem [{\citenamefont {Bombin}\ and\ \citenamefont
  {Martin-Delgado}(2007{\natexlab{a}})}]{bombin:2007:3dcc}%
  \BibitemOpen
  \bibfield  {author} {\bibinfo {author} {\bibfnamefont {H.}~\bibnamefont
  {Bombin}}\ and\ \bibinfo {author} {\bibfnamefont {M.}~\bibnamefont
  {Martin-Delgado}},\ }\href@noop {} {\bibfield  {journal} {\bibinfo  {journal}
  {Phys. Rev. Lett.}\ }\textbf {\bibinfo {volume} {98}},\ \bibinfo {pages}
  {160502} (\bibinfo {year} {2007}{\natexlab{a}})}\BibitemShut {NoStop}%
\bibitem [{\citenamefont {Bombin}(2015{\natexlab{b}})}]{bombin:2015:gauge}%
  \BibitemOpen
  \bibfield  {author} {\bibinfo {author} {\bibfnamefont {H.}~\bibnamefont
  {Bombin}},\ }\href@noop {} {\bibfield  {journal} {\bibinfo  {journal} {New J.
  Phys.}\ }\textbf {\bibinfo {volume} {17}},\ \bibinfo {pages} {083002}
  (\bibinfo {year} {2015}{\natexlab{b}})}\BibitemShut {NoStop}%
\bibitem [{\citenamefont {Bombin}\ and\ \citenamefont
  {Martin-Delgado}(2006)}]{bombin:2006:2dcc}%
  \BibitemOpen
  \bibfield  {author} {\bibinfo {author} {\bibfnamefont {H.}~\bibnamefont
  {Bombin}}\ and\ \bibinfo {author} {\bibfnamefont {M.}~\bibnamefont
  {Martin-Delgado}},\ }\href@noop {} {\bibfield  {journal} {\bibinfo  {journal}
  {Phys. Rev. Lett.}\ }\textbf {\bibinfo {volume} {97}},\ \bibinfo {pages}
  {180501} (\bibinfo {year} {2006})}\BibitemShut {NoStop}%
\bibitem [{\citenamefont {Bombin}\ and\ \citenamefont
  {Martin-Delgado}(2007{\natexlab{b}})}]{bombin:2007:branyons}%
  \BibitemOpen
  \bibfield  {author} {\bibinfo {author} {\bibfnamefont {H.}~\bibnamefont
  {Bombin}}\ and\ \bibinfo {author} {\bibfnamefont {M.}~\bibnamefont
  {Martin-Delgado}},\ }\href@noop {} {\bibfield  {journal} {\bibinfo  {journal}
  {Phys. Rev. B}\ }\textbf {\bibinfo {volume} {75}},\ \bibinfo {pages} {75103}
  (\bibinfo {year} {2007}{\natexlab{b}})}\BibitemShut {NoStop}%
\bibitem [{\citenamefont {Bombin}()}]{bombin:2018:colorful}%
  \BibitemOpen
  \bibfield  {author} {\bibinfo {author} {\bibfnamefont {H.}~\bibnamefont
  {Bombin}},\ }\href@noop {} {\enquote {\bibinfo {title} {2d quantum
  computation with 3d topological codes},}\ }\bibinfo {note} {ArXiv:
  quant-ph/1810.?????}\BibitemShut {Stop}%
\bibitem [{\citenamefont {Rudolph}(2017)}]{rudolph:2017:optimistic}%
  \BibitemOpen
  \bibfield  {author} {\bibinfo {author} {\bibfnamefont {T.}~\bibnamefont
  {Rudolph}},\ }\href@noop {} {\bibfield  {journal} {\bibinfo  {journal} {APL
  Photonics}\ }\textbf {\bibinfo {volume} {2}},\ \bibinfo {pages} {030901}
  (\bibinfo {year} {2017})}\BibitemShut {NoStop}%
\bibitem [{\citenamefont {Bombin}(2016)}]{bombin:2016:dimensional}%
  \BibitemOpen
  \bibfield  {author} {\bibinfo {author} {\bibfnamefont {H.}~\bibnamefont
  {Bombin}},\ }\href@noop {} {\bibfield  {journal} {\bibinfo  {journal} {New J.
  Phys.}\ }\textbf {\bibinfo {volume} {18}},\ \bibinfo {pages} {043038}
  (\bibinfo {year} {2016})}\BibitemShut {NoStop}%
\bibitem [{\citenamefont {Knill}\ \emph {et~al.}(1996)\citenamefont {Knill},
  \citenamefont {Laflamme},\ and\ \citenamefont
  {Zurek}}]{knill:1996:threshold}%
  \BibitemOpen
  \bibfield  {author} {\bibinfo {author} {\bibfnamefont {E.}~\bibnamefont
  {Knill}}, \bibinfo {author} {\bibfnamefont {R.}~\bibnamefont {Laflamme}}, \
  and\ \bibinfo {author} {\bibfnamefont {W.}~\bibnamefont {Zurek}},\
  }\href@noop {} {\bibfield  {journal} {\bibinfo  {journal} {arXiv:
  quant-ph/9610011}\ } (\bibinfo {year} {1996})}\BibitemShut {NoStop}%
\bibitem [{\citenamefont {Beverland}\ \emph {et~al.}(2016)\citenamefont
  {Beverland}, \citenamefont {Buerschaper}, \citenamefont {Koenig},
  \citenamefont {Pastawski}, \citenamefont {Preskill},\ and\ \citenamefont
  {Sijher}}]{beverland:2016:protected}%
  \BibitemOpen
  \bibfield  {author} {\bibinfo {author} {\bibfnamefont {M.}~\bibnamefont
  {Beverland}}, \bibinfo {author} {\bibfnamefont {O.}~\bibnamefont
  {Buerschaper}}, \bibinfo {author} {\bibfnamefont {R.}~\bibnamefont {Koenig}},
  \bibinfo {author} {\bibfnamefont {F.}~\bibnamefont {Pastawski}}, \bibinfo
  {author} {\bibfnamefont {J.}~\bibnamefont {Preskill}}, \ and\ \bibinfo
  {author} {\bibfnamefont {S.}~\bibnamefont {Sijher}},\ }\href@noop {}
  {\bibfield  {journal} {\bibinfo  {journal} {Journal of Mathematical Physics}\
  }\textbf {\bibinfo {volume} {57}},\ \bibinfo {pages} {022201} (\bibinfo
  {year} {2016})}\BibitemShut {NoStop}%
\end{thebibliography}%

\end{document}